\documentclass[11pt]{article}
\usepackage{float}
\usepackage{subcaption}
\usepackage{amsmath,amssymb,amsthm,mathtools}
\usepackage{bbm}
\usepackage{geometry}
\usepackage{hyperref}
\usepackage{enumitem}
\usepackage{microtype}
\usepackage{graphicx}
\usepackage{placeins}
\usepackage[numbers,sort&compress]{natbib}
\usepackage{booktabs}
\usepackage{multirow}
\usepackage{siunitx}
\usepackage{tikz}
\usepackage{pgfplots}
\usepgfplotslibrary{groupplots}
\pgfplotsset{compat=1.18}
\pgfplotsset{
  expplot/.style={
    width=8.65cm,
    height=7.05cm,
    axis lines=left,
    axis line style={-},
    clip=false,
    tick align=outside,
    tick pos=left,
    ymin=0,
    ymax=0.52,
    ytick={0,0.1,0.2,0.3,0.4,0.5},
    ylabel={Power},
    grid=major,
    major grid style={gray!25, dash pattern=on 1.1pt off 1.3pt, line width=0.25pt},
    xlabel style={font=\large},
    ylabel style={font=\large},
    tick label style={font=\small,/pgf/number format/1000 sep={}},
    legend style={
      font=\scriptsize,
      draw=gray!35,
      fill=white,
      fill opacity=0.96,
      text opacity=1,
      cells={anchor=west},
      inner sep=2.2pt,
      rounded corners=0pt,
      at={(0.98,0.03)},
      anchor=south east,
    },
  },
  series/.style={
    line width=0.9pt,
    mark size=2.15pt,
    mark options={solid, draw=., fill=white, line width=0.75pt},
  },
  bench/.style={black, dashed, dash pattern=on 5pt off 3pt, line width=0.7pt},
}
\hypersetup{colorlinks=true,linkcolor=blue,urlcolor=blue,citecolor=blue}

\newcommand{\mP}{\mathbb{P}}
\newcommand{\E}{\mathbb{E}}
\newcommand{\Var}{\mathrm{Var}}
\newcommand{\Cov}{\mathrm{Cov}}
\newcommand{\ind}{\mathbbm{1}}
\newcommand{\iid}{\overset{\mathrm{iid}}{\sim}}
\newcommand{\toP}{\xrightarrow{P}}
\newcommand{\dto}{\xrightarrow{D}}
\newcommand{\R}{\mathbb{R}}
\newcommand{\N}{\mathcal{N}}
\newcommand{\Sn}{\mathfrak{S}_n}

\newtheorem{theorem}{Theorem}
\newtheorem{lemma}{Lemma}
\newtheorem{proposition}{Proposition}

\newtheorem{remark}{Remark}
\newtheorem{definition}{Definition}
\newtheorem{assumption}{Assumption}

\title{When Do Generalized Permutation Tests Achieve Optimal Power? \\ A Dispersion Characterization}
\author{Yongmin Kim\thanks{Korea Advanced Institute of Science and Technology (KAIST).
  E-mail: \texttt{junmini5@kaist.ac.kr}}
  \and
  Ilmun Kim\thanks{Korea Advanced Institute of Science and Technology (KAIST).
  E-mail: \texttt{ilmunk@kaist.ac.kr}}}

\begin{document}
\maketitle

\begin{abstract}
We study generalized Monte Carlo permutation tests under a non-uniform distribution on permutations. Focusing on the difference-in-means statistic, we introduce two scalar dispersion measures that quantify departures from complete randomization at the individual and pairwise levels. We show that if both dispersions vanish asymptotically, then the conditional permutation distribution converges to its Gaussian benchmark, the critical value stabilizes, and the test attains optimal Pitman local power. Conversely, if these dispersions fail to vanish, the permutation distribution does not self-average, the critical value need not stabilize, and optimal local power cannot in general be guaranteed. We further show that beyond the standard Pitman local model, suitably chosen non-uniform permutation distributions can strictly dominate the uniform distribution by exploiting nuisance structure in the data.
\end{abstract}

\section{Introduction}

Permutation tests are among the oldest and most widely used nonparametric tools in statistics.
Introduced by \citet{fisher1935design} and \citet{pitman1937,pitman1938},
they construct a reference distribution by re-randomizing the treatment labels
and comparing the observed test statistic against this distribution.
When the data are exchangeable under the null hypothesis,
the resulting $p$-value controls the type~I error at the nominal level in finite samples,
without any parametric modeling assumptions on the data-generating process.

The classical theory of permutation tests focuses almost exclusively on the case where the
re-randomization is drawn uniformly from the symmetric group~$\Sn$, or equivalently,
from a simple random sample of treatment assignments.
In this setting, the asymptotic behavior of the permutation distribution is well understood:
\citet{hoeffding1952} established that the large-sample power of uniform permutation tests
matches that of the corresponding parametric tests under local alternatives,
and \citet{hajek1961,hajek1999theory} provided the permutation central limit theorem
that underpins this result.
These foundational results, further developed by
\citet{lehmann2005testing} and \citet{romano1989,romano1990},
show that uniform permutation tests are asymptotically optimal in the Pitman efficiency sense
for a broad class of statistics, including linear rank statistics and studentized test statistics.
More recent work has extended these ideas in several directions:
non-asymptotic power analysis~\citep{berrett2021optimal,kim2022minimax,schrab2023mmd},
conditional exchangeability~\citep{berrett2020conditional,kim2022local},
and robustness to exchangeability violations~\citep{chung2013exact,schrab2024robust}.

A recent paper by \citet{ramdas2023permutation}
substantially expands the scope of permutation testing.
They show that finite-sample validity is preserved
even when the Monte Carlo permutations are drawn from an arbitrary distribution~$q_n$ on~$\Sn$,
which need not have any group structure at all.
This opens the door to a much larger class of generalized permutation tests,
with potential computational and statistical benefits.
However, their analysis addresses only validity;
\emph{the power properties of generalized permutation tests with non-uniform~$q_n$
remain open.}

This gap raises two natural questions, operating at different levels.
The first concerns the standard Pitman local model with i.i.d.\ noise.
Under what conditions on~$q_n$ does the generalized permutation test
attain optimal asymptotic power?
The second concerns what happens when the data carry additional structure beyond the Pitman local model.
Can a non-uniform choice of~$q_n$ that exploits this structure
yield strictly higher power than the uniform distribution on~$\Sn$,
even when both choices would be optimal in the Pitman local setting?

In this paper, we provide a complete answer to both questions
in the two-sample testing framework.
For the first question, we introduce two scalar quantities,
the \emph{first-order dispersion}~$V_{1,n}$ and the \emph{second-order dispersion}~$V_{2,n}$,
that measure how far the marginal and pairwise assignment probabilities
induced by~$q_n$ deviate from those of complete randomization.
Under the Pitman local model, we show that the joint condition
$V_{1,n}\to 0$ and $V_{2,n}\to 0$ is sufficient for the conditional permutation distribution
to converge to the standard normal, the critical values to stabilize at $z_{1-\alpha}$,
and the test to attain optimal Pitman local power (Theorem~\ref{thm:power}).
Conversely, under Gaussian noise, if $V_{1,n}+V_{2,n}\not\to 0$ then the permutation distribution
fails to self-average and the critical values cannot stabilize at the $z$-test benchmark
at every nominal level (Theorem~\ref{thm:necessary}).
Thus, within the Gaussian-noise Pitman local model,
the two dispersions provide a necessary and sufficient characterization of when
a generalized permutation test achieves optimal local power.

For the second question, we step outside the Pitman local model.
Specifically, we consider a stratified two-sample design
in which the data exhibit between-stratum heterogeneity
on top of a local treatment effect.
In this setting, the dispersion framework alone no longer determines power.
Both the full permutation ($q_n=\mathrm{Uniform}(\Sn)$)
and a within-stratum permutation satisfy $V_{1,n}=0$ and $V_{2,n}\to 0$,
yet they yield different critical values because their permutation variances
absorb different amounts of the between-stratum variation.
We show that the within-stratum design achieves strictly higher power
(Proposition~\ref{prop:power-compare}),
recovering in the generalized permutation framework the classical variance-reduction benefit
of stratification.
Together, the two results delineate the reach of the dispersion conditions.
They pin down power completely within the Pitman local model,
but richer data structure calls for choices of $q_n$ that respect that structure.

\subsection{Summary of contributions}\label{sec:contributions}

Our main contributions are as follows.

\begin{enumerate}[label=(\roman*)]
\item \textbf{Dispersion characterization of optimal power.}
Within the Pitman local model,
we introduce the first- and second-order dispersions $(V_{1,n},V_{2,n})$
and show that they provide a necessary and sufficient characterization of when
a generalized permutation test attains optimal Pitman local power\footnote{Throughout, ``optimal local power'' refers to the Pitman local power of the $z$-test, which is asymptotically optimal for location alternatives by Le Cam's local asymptotic normality; see \citet[][Chapter~15]{vandervaart1998asymptotic}.}
for the difference-in-means statistic:
vanishing dispersions are sufficient (Theorem~\ref{thm:power}),
and under Gaussian noise, non-vanishing dispersions imply failure (Theorem~\ref{thm:necessary}).
We illustrate the sharpness of these conditions through three canonical examples:
the full symmetric group ($V_{1,n}=V_{2,n}=0$, optimal),
the cyclic group ($V_{1,n}=0$ but $V_{2,n}\not\to 0$, non-optimal),
and a block permutation design ($V_{1,n}=0$, $V_{2,n}=O(n^{-2})\to 0$, optimal).
These examples demonstrate, in particular, that marginal balance ($V_{1,n}=0$) alone is not sufficient
and that pairwise regularity ($V_{2,n}\to 0$) is also necessary.

\item \textbf{Extension to asymptotically linear statistics.}
We extend the dispersion framework to statistics admitting a permutation-typical
asymptotic linearization (Theorem~\ref{thm:AL}).
This covers the standard studentized difference in means once the required $q_n$-typical variance consistency is verified.

\item \textbf{Power gains beyond the Pitman local model.}
Going beyond the Pitman local model, we consider a stratified design
in which between-stratum heterogeneity coexists with a local treatment effect.
In this setting, the dispersion framework no longer determines power:
both the full permutation and the within-stratum permutation satisfy 
$V_{1,n}=0$ and $V_{2,n}\to 0$,
yet the latter achieves strictly higher power
by avoiding the absorption of between-stratum variation into the permutation variance
(Proposition~\ref{prop:power-compare}).
The power of the full permutation test degrades as the between-stratum variation grows,
while the within-stratum test is unaffected.
\end{enumerate}

Taken together, these results offer a unified message: within the Pitman local model,
what determines whether a generalized permutation test is optimally powerful is precisely
how close its randomization law is to the uniform distribution,
and the dispersions $(V_{1,n},V_{2,n})$ provide an exact quantitative measure of that proximity.
Outside this local regime, however, deliberately departing from uniformity
by exploiting known nuisance structure can yield additional power
that the dispersion framework alone does not capture.

\subsection{Organization}\label{sec:notation}

The remainder of this paper is organized as follows.
Section~\ref{sec:framework} introduces the framework:
the two-sample design, the generalized permutation randomization of~\cite{ramdas2023permutation},
and the Pitman local model.
Section~\ref{sec:optimal} presents the main results on optimal power under vanishing dispersions,
including the sufficient conditions (Theorems~\ref{thm:cond-normal}--\ref{thm:power}),
the necessary conditions (Theorem~\ref{thm:necessary}),
the extension to asymptotically linear statistics (Theorem~\ref{thm:AL}),
and illustrative examples.
Section~\ref{sec:nonuniform} steps outside the Pitman local model and demonstrates
that non-uniform randomization can strictly dominate the uniform distribution in terms of power
in the presence of between-stratum heterogeneity.
Section~\ref{sec:simulations} provides simulation evidence.
Section~\ref{sec:conclusion} concludes.
Proofs are collected in the Appendix.

\section{Framework}\label{sec:framework}

\subsection{Setup}\label{sec:setup}

Let $[n]=\{1,\dots,n\}$ index the experimental units and let $\Sn$ denote the symmetric group on $[n]$.
We consider a two-sample design with $n_1$ treated units and $n_0$ control units, where $n_1+n_0=n$ and $\rho_n:=n_1/n$.
We observe a single outcome vector $X=(X_1,\dots,X_n)^\top\in\R^n$.
An assignment vector $W\in\{0,1\}^n$ encodes the treatment labels, where $W_i=1$ indicates treatment and $W_i=0$ indicates control; throughout we restrict attention to assignments with exactly $n_1$ treated units,
so $\sum_{i=1}^n W_i=n_1$.

For notational convenience, we fix a reference labeling in which the first $n_1$ indices are treated and the remaining $n_0$ indices are controls. Let
\[
w_0 := (\underbrace{1,\dots,1}_{n_1},\underbrace{0,\dots,0}_{n_0})^\top,\qquad A_n:=\{1,\dots,n_1\},
\]
so that $w_{0,i}=\ind\{i\in A_n\}$.

\begin{assumption}[Asymptotic proportions]\label{ass:rhon}
$\rho_n\to\rho\in(0,1)$ and $n_1,n_0\to\infty$.
\end{assumption}

This ensures that both groups grow without bound and that the treatment fraction remains bounded away from~$0$ and~$1$, conditions under which the central limit arguments in later sections apply.

\subsection{Generalized permutation randomization}\label{sec:gen-perm}

The generalized permutation test of \citet{ramdas2023permutation} allows the Monte Carlo
permutations to be drawn from an arbitrary distribution $q_n$ on $\Sn$.
Concretely, one draws a base permutation $\pi_0\sim q_n$ and Monte Carlo permutations
$\pi_1,\dots,\pi_M\iid q_n$, all independently of the data.

The base draw $\pi_0$ is essential when $q_n$ is non-uniform: the effective randomization
applied to the data is not $\pi_m$ itself but the right-translate
\[
  \tau := \pi_m \circ \pi_0^{-1}, \qquad \pi_m \sim q_n,
\]
so that $\tau\sim q_n^{(\pi_0)}$ where $q_n^{(s)}(t):=q_n(t\circ s)$.
This translation aligns the Monte Carlo draws with the observed assignment $w_0$
and thereby restores the exchangeability needed for finite-sample validity.
When $q_n$ is uniform, right-translation preserves the law, so $\pi_0$ plays no role
and one may draw $\pi_1,\dots,\pi_M$ directly from $q_n$.

Each permutation $\tau$ induces an assignment vector $W(\tau)\in\{0,1\}^n$ via
\[
  W_i(\tau) := \ind\{\tau(i)\in A_n\} = \ind\{\tau(i)\le n_1\}, \qquad i=1,\dots,n,
\]
which always has exactly $n_1$ treated units.
Since all test statistics below depend on $\tau$ only through $W(\tau)$,
the many-to-one nature of the map $\tau\mapsto W(\tau)$ is inconsequential.

\paragraph{Test statistic (difference in means).}
For any assignment vector $W\in\{0,1\}^n$ with $\sum_{i=1}^n W_i=n_1$,
define the (standardized) difference-in-means statistic
\begin{equation}\label{eq:Tn}
T_n(X,W)
:=
\frac{1}{\sigma_\varepsilon\sqrt{n\rho_n(1-\rho_n)}}\sum_{i=1}^n (W_i-\rho_n)X_i,
\end{equation}
where $\sigma_\varepsilon^2 > 0$ denotes the noise variance (defined formally in Assumption~\ref{ass:pitman} below).
When $W=w_0$, we write $T_n^{\mathrm{obs}}:=T_n(X,w_0)$.
Since permutation tests are scale-invariant (the $p$-value depends on $T_n$ only through
its rank among the permuted values), the value of $\sigma_\varepsilon$ is irrelevant to the test itself.
The oracle normalization is adopted purely to simplify the asymptotic statements,
and Section~\ref{sec:AL} gives conditions under which studentized and other
asymptotically linear statistics inherit the same limiting behavior.
Throughout, we focus on the one-sided test that rejects for large values of $T_n^{\mathrm{obs}}$,
and the extension to two-sided tests follows naturally.

\paragraph{Validity.}
The following result, due to \citet[][Theorem 2]{ramdas2023permutation}, confirms that
this construction yields a valid $p$-value for any choice of $q_n$.

\begin{theorem}[Validity with Monte Carlo draws from an arbitrary permutation distribution]
\label{thm:ramdas-valid}
Let $q_n$ be any pmf on $\Sn$. 
Draw $\pi_0,\pi_1,\dots,\pi_M\iid q_n$ independently of the data.
Define
\[
P_n(X)
:=
\frac{
1+\sum_{m=1}^M 
\ind\big\{
T_n(X,W(\pi_m\circ\pi_0^{-1}))
\ge 
T_n(X,w_0)
\big\}
}{1+M}.
\]
If $X$ is exchangeable under the null, then $P_n(X)$ is a valid $p$-value, i.e.,
\[
\mP\big(P_n(X)\le \alpha\big)\le \alpha
\quad\text{for all }\alpha\in[0,1].
\]
\end{theorem}

\subsection{Pitman local model}\label{sec:pitman}

We analyze power under the standard Pitman local alternative,
which posits an order-$n^{-1/2}$ mean shift on the treated units.

\begin{assumption}[Pitman local alternative]\label{ass:pitman}
Let $\varepsilon_i$ be i.i.d.\ with $\E[\varepsilon_i]=0$,
$\Var(\varepsilon_i)=\sigma_\varepsilon^2\in(0,\infty)$,
and $\E|\varepsilon_i|^{2+\kappa}<\infty$ for some $\kappa>0$.
Assume
\[
X_i = \mu + \frac{\delta}{\sqrt{n}}\,w_{0,i} + \varepsilon_i,\qquad i=1,\dots,n,
\]
for some fixed $\mu\in\R$ and fixed local effect $\delta\in\R$.
\end{assumption}

Under the null hypothesis $H_0:\delta=0$, the data $X_1,\dots,X_n$ are i.i.d.\ with
common mean $\mu$ and hence exchangeable, so Theorem~\ref{thm:ramdas-valid} ensures validity.
Under the local alternative $H_1:\delta\neq 0$, the treated units $\{i:w_{0,i}=1\}$
carry an order-$n^{-1/2}$ mean shift, and our analysis characterizes the resulting asymptotic power.
The intercept $\mu$ cancels in $T_n$ since $\sum_{i=1}^n(w_{0,i}-\rho_n)=0$,
and is included only to make explicit that the test is invariant to common location shifts.

\section{Optimal Pitman power under vanishing dispersions}\label{sec:optimal}

Throughout this section, we analyze the population version of the generalized permutation test,
in which the critical value $c_n(\alpha\mid X,\pi_0)$ is defined as the exact $(1-\alpha)$-quantile
of the conditional permutation distribution under $q_n^{(\pi_0)}$.
This corresponds to the limit $M\to\infty$ of the Monte Carlo test of Theorem~\ref{thm:ramdas-valid}.
For finite $M$, the Monte Carlo critical value $\hat{c}_n(\alpha\mid X,\pi_0,M)$ differs from
the population critical value $c_n(\alpha\mid X,\pi_0)$ by $O_P(M^{-1/2})$,
as follows from a standard CLT for empirical quantiles applied to the $M$ i.i.d.\ permuted statistics.
Since this error vanishes as $M\to\infty$ independently of $n$,
it does not interact with the $n\to\infty$ asymptotics studied here,
and all power statements below carry over to the Monte Carlo test for any $M$ that tends to infinity.

\subsection{First- and second-order dispersions}\label{sec:dispersions}

To quantify how far $q_n$ departs from complete randomization,
we introduce two scalar dispersion measures based on the marginal and pairwise
assignment probabilities induced by $q_n$.
Let $\pi\sim q_n$, and define the one- and two-point block probabilities
\[
a_k:=\mP(\pi(k)\in A_n),\qquad b_{k\ell}:=\mP(\pi(k)\in A_n,\ \pi(\ell)\in A_n),\ \ k\neq \ell,
\]
and the complete-randomization benchmark
\[
\tau_n:=\frac{n_1(n_1-1)}{n(n-1)}.
\]

\begin{definition}[$q_n$-dispersions]\label{def:V}
Define
\[
V_{1,n}:=\frac1n\sum_{k=1}^n (a_k-\rho_n)^2,
\qquad V_{2,n}:=\frac1{n^2}\sum_{k\neq \ell}(b_{k\ell}-\tau_n)^2.
\]
\end{definition}

The first-order dispersion $V_{1,n}$ measures how far the marginal treatment assignment probabilities $a_k=\mP(\pi(k)\in A_n)$ deviate from the complete-randomization benchmark $\rho_n=n_1/n$.
Under complete randomization, every unit is equally likely to be treated, so $a_k=\rho_n$ for all $k$ and $V_{1,n}=0$.
A nonzero $V_{1,n}$ indicates that some units are systematically more (or less) likely to receive treatment than others.
The second-order dispersion $V_{2,n}$ captures pairwise dependence,
measuring how far the joint probabilities $b_{k\ell}=\mP(\pi(k),\pi(\ell)\in A_n)$ deviate from the benchmark $\tau_n=n_1(n_1-1)/(n(n-1))$.
Even when marginal probabilities are balanced ($V_{1,n}=0$),
the assignment of one unit to treatment may strongly predict the assignment of another,
and this pairwise dependence is reflected in $V_{2,n}$.
As we show below, both dispersions must vanish for optimal power; marginal balance alone is not sufficient.

\begin{remark}[Invariance under right-translation]
The dispersions $(V_{1,n},V_{2,n})$ depend on $q_n$ only through symmetric averages of
$(a_k,b_{k\ell})$ and are therefore invariant under right-translation:
for every $s\in\Sn$, the law $q_n^{(s)}$ yields the same dispersions as $q_n$.
Consequently, our dispersion-based analysis applies uniformly across the family of
conditional laws $\{q_n^{(\pi_0)}:\pi_0\in\Sn\}$ arising from the base draw~$\pi_0$.
See Lemma~\ref{lem:translate} in the Appendix for a formal statement.
\end{remark}

We formalize the vanishing condition as follows.

\begin{assumption}[Vanishing dispersions]\label{ass:V}
Assume $V_{1,n}\to 0$ and $V_{2,n}\to 0$.
\end{assumption}

\subsection{Main results}\label{sec:main-results}

Throughout this subsection, let $\tau\sim q_n^{(\pi_0)}$ be drawn independently of $X$
conditional on $\pi_0$, and define the conditional permutation cumulative distribution function (cdf)
\[
F_n(t\mid X,\pi_0):=\mP\!\left(T_n(X,W(\tau))\le t\mid X,\pi_0\right)
\]
with $(1-\alpha)$-quantile $c_n(\alpha\mid X,\pi_0)$.
Under vanishing dispersions, $F_n$ converges uniformly to $\Phi$,
the critical values stabilize at $z_{1-\alpha}$,
and the generalized permutation test attains the same local power as the $z$-test.

\begin{theorem}[Permutation CLT under vanishing dispersions]\label{thm:cond-normal}
Under Assumptions~\ref{ass:rhon}, \ref{ass:pitman}, and \ref{ass:V},
\[
\sup_{t\in\R}\left|F_n(t\mid X,\pi_0)-\Phi(t)\right|\toP 0,
\]
and consequently
\[
c_n(\alpha\mid X,\pi_0)\toP z_{1-\alpha}.
\]
\end{theorem}
\begin{proof}
    See Appendix~\ref{app:proof-cond-normal}.
\end{proof}
The proof proceeds in three steps.
The first step decomposes the resampled statistic $T_n=T_n(X,W(\tau))$ into a noise term $N_n$
and a signal term $L_n$ using the Pitman local structure of the data.
The signal term involves the \emph{reference overlap} $O_n^{(0)}:=\frac1n\sum_{i\in A_n}W_i(\tau)$,
which measures how much the resampled assignment agrees with the original labeling.
Vanishing dispersions force $O_n^{(0)}\toP\rho^2$
(Proposition~\ref{prop:overlap} in the Appendix),
which makes $L_n$ negligible.

The second step establishes a conditional bivariate CLT for the noise pair $(N_n,N_n')$
arising from two independent resampled assignments.
Conditional on the assignments, the noise terms are sums of independent random variables
(driven by the i.i.d.\ errors $\varepsilon_i$) whose correlation is
$\gamma_n=(O_n-\rho_n^2)/(\rho_n(1-\rho_n))$,
where $O_n:=\frac1n\sum_i W_iW_i'$ is the overlap between the two resampled assignments.
Vanishing dispersions again force $O_n\toP\rho^2$, so $\gamma_n\toP 0$,
and the noise pair becomes asymptotically i.i.d.\ standard normal.

The third step applies a two-randomization criterion
(Lemma~\ref{lem:two-rand} in the Appendix):
roughly, if two conditionally i.i.d.\ copies of a statistic are jointly Gaussian in the limit,
then the conditional cdf must concentrate uniformly around that Gaussian.
This upgrades the joint CLT from Step~2 to the desired uniform convergence of $F_n$.
Together, the three steps yield Theorem~\ref{thm:cond-normal}.

\begin{proposition}[Observed statistic has a local mean shift]\label{prop:obs}
Let $\Delta := \delta\sqrt{\rho(1-\rho)}/\sigma_\varepsilon$ denote the asymptotic signal-to-noise ratio. Under Assumptions~\ref{ass:rhon} and~\ref{ass:pitman},
\[
T_n^{\mathrm{obs}} \dto \N(\Delta,1).
\]
\end{proposition}

\begin{proof}
See Appendix~\ref{app:proof-obs}.
\end{proof}

The following theorem is an immediate consequence.

\begin{theorem}[Vanishing dispersions yield optimal local power]\label{thm:power}
Under Assumptions~\ref{ass:rhon}, \ref{ass:pitman}, and \ref{ass:V},
$c_n(\alpha\mid X,\pi_0)\toP z_{1-\alpha}$ and
\[
\mP\bigl(T_n^{\mathrm{obs}} > c_n(\alpha\mid X,\pi_0)\bigr)
\ \rightarrow\
1-\Phi\!\left(z_{1-\alpha}-\Delta\right).
\]
In particular, the generalized permutation test attains the asymptotically optimal Pitman local power.
\end{theorem}

\begin{proof}
Theorem~\ref{thm:cond-normal} gives $c_n(\alpha\mid X,\pi_0)\toP z_{1-\alpha}$
and Proposition~\ref{prop:obs} gives $T_n^{\mathrm{obs}}\dto \N(\Delta,1)$.
Slutsky's theorem then gives
$T_n^{\mathrm{obs}}-c_n(\alpha\mid X,\pi_0)\dto \N(\Delta-z_{1-\alpha},1)$,
and the result follows.
\end{proof}

\subsection{Necessary conditions: convergence failure}\label{sec:necessary}

The previous subsection established that vanishing dispersions are sufficient for optimal power.
A natural question is whether they are also necessary.
The following theorem, proved under the additional assumption of Gaussian noise,
shows that they are: if $V_{1,n}+V_{2,n}\not\to 0$,
the permutation critical values cannot stabilize at the $z$-test benchmark.

\begin{assumption}[Gaussian noise]\label{ass:gauss-noise}
In Assumption~\ref{ass:pitman}, assume additionally that $\varepsilon_i\iid \N(0,\sigma_\varepsilon^2)$.
\end{assumption}

\begin{theorem}[Convergence failure]\label{thm:necessary}
Under Assumptions~\ref{ass:rhon}, \ref{ass:pitman}, and~\ref{ass:gauss-noise},
if $V_{1,n}+V_{2,n}\not\to 0$, then there exists $\alpha\in(0,1)$ such that
\[
 c_n(\alpha\mid X,\pi_0)\ \not\toP\ z_{1-\alpha}.
\]
Consequently, without vanishing dispersions there is no general guarantee that the generalized permutation test attains the $z$-test Pitman local-power benchmark at every nominal level.
\end{theorem}

\begin{proof}
See Appendix~\ref{app:proof-necessary}.
\end{proof}

The proof establishes a non-self-averaging phenomenon.
Recall that the overlap $O_n:=\frac{1}{n}\sum_{i=1}^n W_i W_i'$ measures agreement between two independent
resampled assignments, and the reference overlap $O_n^{(0)}:=\frac{1}{n}\sum_{i\in A_n}W_i(\tau)$
measures agreement with the observed labeling.
The conditional permutation distribution depends on $\pi_0$ through these quantities,
and when the dispersions do not vanish, they retain nontrivial fluctuations across realizations of $\pi_0$
that propagate into the conditional distribution of $T_n$.

More precisely, if $F_n(\cdot\mid X,\pi_0)$ were to converge in probability to some deterministic cdf $F$,
then $\E[T_n^2\mid X,\pi_0]$ would also have to concentrate.
Under Gaussian noise, the conditional law of $T_n$ given $(W,W',\pi_0)$ is Gaussian with
uniformly bounded variance, so $T_n^2$ is uniformly integrable, and the fourth-moment structure
of bivariate normals yields the identity
\[
\Var\bigl(\E[T_n^2\mid X,\pi_0]\bigr)\ge 2\,\E[\gamma_n^2],
\]
where $\gamma_n=(O_n-\rho_n^2)/(\rho_n(1-\rho_n))$ is the overlap correlation.
Lemma~\ref{lem:overlap-V} in the Appendix links $\E[\gamma_n^2]$ back to the dispersions:
if $V_{1,n}+V_{2,n}\not\to 0$, then $\E[\gamma_n^2]\not\to 0$,
so the conditional second moment fluctuates and $F_n$ cannot converge to any deterministic limit.

\subsection{Extension to asymptotically linear statistics}\label{sec:AL}

The preceding results are stated for the oracle-normalized difference in means.
The same dispersion argument also applies to statistics that are asymptotically linear,
but the relevant linearization must be strong enough for a permutation distribution.
A pointwise expansion at the observed assignment is not sufficient, while a supremum over
all balanced assignments is often unnecessarily strong and can fail for otherwise standard procedures.
We therefore formulate the transfer condition on the assignments that are actually sampled from
the conditional law $q_n^{(\pi_0)}$, together with the observed assignment.

Let
\[
\mathcal W_n:=\biggl\{w\in\{0,1\}^n:\ \sum_{i=1}^n w_i=n_1\biggr\}
\]
denote the set of feasible assignment vectors.  For a statistic $S_n(X,w)$, write
$W(\tau)$ for the assignment induced by a permutation $\tau\sim q_n^{(\pi_0)}$.

\begin{assumption}[Asymptotic linearization under the permutation law]\label{ass:AL}
There exist constants $\mu_L\in\R$, $\beta\in\R$, and $\sigma_L^2\in(0,\infty)$,
and i.i.d. random variables $\xi_1,\ldots,\xi_n$ satisfying
$\E[\xi_i]=0$, $\Var(\xi_i)=\sigma_L^2$, and
$\E|\xi_i|^{2+\kappa}<\infty$ for some $\kappa>0$, such that the ideal linear statistic
\[
L_n(w):=
\frac{1}{\sigma_L\sqrt{n\rho_n(1-\rho_n)}}
\sum_{i=1}^n (w_i-\rho_n)
\left(\mu_L+\frac{\beta}{\sqrt n}w_{0,i}+\xi_i\right)
\]
approximates $S_n$ in the following two senses:
for every $\eta>0$,
\begin{equation}\label{eq:qtypical-rem}
\mP_{\tau\sim q_n^{(\pi_0)}}\Bigl(
\bigl|S_n\{X,W(\tau)\}-L_n\{W(\tau)\}\bigr|>\eta
\, \Bigm|\, X,\pi_0\Bigr)\toP 0,
\end{equation}
and
\begin{equation}\label{eq:obs-rem}
\bigl|S_n(X,w_0)-L_n(w_0)\bigr|\toP 0.
\end{equation}
\end{assumption}

Condition~\eqref{eq:qtypical-rem} is deliberately weaker than
$\sup_{w\in\mathcal W_n}|S_n(X,w)-L_n(w)|=o_P(1)$.
It rules out relying on an expansion that holds only at one fixed assignment,
but it does not require control of adversarial assignments that have negligible or zero probability under $q_n^{(\pi_0)}$.
This distinction matters, for example, for studentized statistics: a variance estimator can be consistent
for $q_n^{(\pi_0)}$-typical assignments even though it is not uniformly consistent over all balanced splits of the data.

\begin{theorem}[Generalized permutation tests for asymptotically linear statistics]\label{thm:AL}
Under Assumptions~\ref{ass:rhon}, \ref{ass:V}, and~\ref{ass:AL},
letting $\Delta_L:=\frac{\beta\sqrt{\rho(1-\rho)}}{\sigma_L}$,
the conditional permutation cdf of $S_n(X,W(\tau))$ converges uniformly to $\Phi$,
the critical values satisfy $d_n(\alpha\mid X,\pi_0)\toP z_{1-\alpha}$,
$S_n^{\mathrm{obs}}\dto \N(\Delta_L,1)$, and
\[
\mP\bigl(S_n^{\mathrm{obs}}>d_n(\alpha\mid X,\pi_0)\bigr)
\ \rightarrow\
1-\Phi\!\left(z_{1-\alpha}-\Delta_L\right).
\]
\end{theorem}

\begin{proof}
See Appendix~\ref{app:proof-AL}.
\end{proof}

The theorem is best read as a transfer principle.
Once a statistic is close under the randomization law used by the test to a linear statistic with a Pitman local mean shift,
Theorems~\ref{thm:cond-normal}--\ref{thm:power} apply to the linear statistic, and the approximation error is negligible.
We use this result below for the standard studentized difference in means, where the only additional verification is $q_n$-typical consistency of the variance estimator.

\medskip\noindent
\textbf{Example 1 (studentized difference in means).}
Consider the usual studentized statistic
\[
S_n^{\mathrm{stud}}(X,w)
:=
\frac{\sqrt{n\rho_n(1-\rho_n)}\{\bar X_1(w)-\bar X_0(w)\}}
{\widehat\sigma_p(w)},
\]
where
\[
\bar X_1(w):=\frac1{n_1}\sum_{i=1}^n w_iX_i,
\qquad
\bar X_0(w):=\frac1{n_0}\sum_{i=1}^n(1-w_i)X_i,
\]
and
\[
\widehat\sigma_p^2(w)
:=
\frac1n\sum_{i=1}^n w_i\{X_i-\bar X_1(w)\}^2
+
\frac1n\sum_{i=1}^n(1-w_i)\{X_i-\bar X_0(w)\}^2.
\]
The particular finite-sample denominator convention is immaterial for the following asymptotic argument.

Under Assumption~\ref{ass:pitman}, suppose additionally that $\E\varepsilon_i^4<\infty$.
If the base law $q_n$ satisfies $V_{1,n}\to0$ and $V_{2,n}\to0$, then the same two-moment calculation underlying Theorem~\ref{thm:cond-normal} gives a $q_n$-weighted law of large numbers.  In particular, for $Z_i=X_i$ and $Z_i=X_i^2$,
\[
\mP_{\tau\sim q_n^{(\pi_0)}}\Biggl(
\biggl|\frac1n\sum_{i=1}^n\{W_i(\tau)-\rho_n\}Z_i\biggr|>\eta
\, \Bigm|\, X,\pi_0\Biggr)\toP0
\qquad(\eta>0).
\]
Thus the treatment and control sample means and second moments along $q_n$-typical relabelings converge to their full-sample counterparts.  Since $\E\varepsilon_i^4<\infty$ and the local shift is $O(n^{-1/2})$, this yields
\begin{equation}\label{eq:pooled-var-qtypical}
\mP_{\tau\sim q_n^{(\pi_0)}}\Bigl(
\bigl|\widehat\sigma_p\{W(\tau)\}-\sigma_\varepsilon\bigr|>\eta
\, \Bigm|\, X,\pi_0\Bigr)\toP0
\qquad(\eta>0),
\end{equation}
and the ordinary law of large numbers gives $\widehat\sigma_p(w_0)\toP\sigma_\varepsilon$.
For $w\in\mathcal W_n$, define the oracle-normalized statistic
\[
L_n(w)=
\frac{1}{\sigma_\varepsilon\sqrt{n\rho_n(1-\rho_n)}}
\sum_{i=1}^n(w_i-\rho_n)X_i.
\]
Since
\[
S_n^{\mathrm{stud}}(X,w)
=
\frac{\sigma_\varepsilon}{\widehat\sigma_p(w)}L_n(w),
\]
we have
\[
S_n^{\mathrm{stud}}(X,w)-L_n(w)
=
\left\{
\frac{\sigma_\varepsilon}{\widehat\sigma_p(w)}-1
\right\}L_n(w).
\]
The statistic $L_n(w)$ is $O_P(1)$ under $q_n^{(\pi_0)}$ and at $w_0$.
Therefore~\eqref{eq:pooled-var-qtypical} and $\widehat\sigma_p(w_0)\toP\sigma_\varepsilon$
verify Assumption~\ref{ass:AL} with $\mu_L=0$, $\beta=\delta$, $\sigma_L=\sigma_\varepsilon$, and $\xi_i=\varepsilon_i$.
Consequently, the standard studentized difference in means has the same limiting permutation critical values and the same Pitman power as the oracle statistic.
This treatment is in line with the classical literature on studentized permutation tests, where studentization is handled through the randomization distribution~\citep{janssen1997studentized,chung2013exact}.

\subsection{Illustrative examples of dispersion conditions}\label{sec:examples-dispersion}

To illustrate the sharpness of the dispersion conditions,
we compute $(V_{1,n},V_{2,n})$ for three canonical permutation distributions:
the full symmetric group, the cyclic group, and a block permutation design.
Figure~\ref{fig:dispersion-illustration} visualizes the marginal and pairwise deviation structure
for each design.

These examples demonstrate that $V_{1,n}=0$ (marginal balance) alone is not sufficient for optimal power.
The pairwise condition $V_{2,n}\to 0$ is also essential.

\begin{figure}[ht]
\centering
\begin{tikzpicture}[x=0.42cm,y=0.42cm,
  every node/.style={align=center},
  zero/.style={draw=black!25, fill=black!6, line width=0.3pt},
  diag/.style={draw=black!25, fill=black!22, line width=0.3pt},
  posstrong/.style={draw=black!35, fill=blue!55, line width=0.3pt},
  posweak/.style={draw=black!25, fill=blue!22, line width=0.3pt},
  negstrong/.style={draw=black!35, fill=red!55, line width=0.3pt},
  negweak/.style={draw=black!25, fill=red!22, line width=0.3pt}
]

\begin{scope}[shift={(0,0)}]
  \node[font=\small\bfseries] at (4,2.9) {Uniform on $\mathfrak{S}_n$};
  \node[font=\scriptsize] at (4,2.0) {marginal deviations $a_k-\rho_n$};
  \foreach \k in {1,...,8}{
    \draw[zero] (\k-1,0.85) rectangle ++(1,0.65);
  }
  \node[font=\scriptsize, anchor=east] at (-0.2,1.18) {$V_1$};
  \node[font=\scriptsize] at (4,0.2) {pairwise deviations $b_{k\ell}-\tau_n$};
  \foreach \i in {1,...,8}{
    \foreach \j in {1,...,8}{
      \ifnum\i=\j
        \draw[diag] (\j-1,-\i-1) rectangle ++(1,1);
      \else
        \draw[zero] (\j-1,-\i-1) rectangle ++(1,1);
      \fi
    }
  }
  \draw[black!45, line width=0.6pt] (0,-9) rectangle (8,-1);
  \node[font=\small] at (4,-10.5)
    {$V_{1,n}=0,\quad V_{2,n}=0$\\\textit{optimal power}};
\end{scope}

\begin{scope}[shift={(11.5,0)}]
  \node[font=\small\bfseries] at (4,2.9) {Cyclic shifts};
  \node[font=\scriptsize] at (4,2.0) {marginal deviations $a_k-\rho_n$};
  \foreach \k in {1,...,8}{
    \draw[zero] (\k-1,0.85) rectangle ++(1,0.65);
  }
  \node[font=\scriptsize, anchor=east] at (-0.2,1.18) {$V_1$};
  \node[font=\scriptsize] at (4,0.2) {pairwise deviations $b_{k\ell}-\tau_n$};
  \foreach \i in {1,...,8}{
    \foreach \j in {1,...,8}{
      \ifnum\i=\j
        \draw[diag] (\j-1,-\i-1) rectangle ++(1,1);
      \else
        \pgfmathtruncatemacro{\dd}{abs(\i-\j)}
        \pgfmathtruncatemacro{\dc}{min(\dd,8-\dd)}
        \ifnum\dc=1
          \draw[posstrong] (\j-1,-\i-1) rectangle ++(1,1);
        \else\ifnum\dc=2
          \draw[posweak] (\j-1,-\i-1) rectangle ++(1,1);
        \else\ifnum\dc=3
          \draw[negweak] (\j-1,-\i-1) rectangle ++(1,1);
        \else
          \draw[negstrong] (\j-1,-\i-1) rectangle ++(1,1);
        \fi\fi\fi
      \fi
    }
  }
  \draw[black!45, line width=0.6pt] (0,-9) rectangle (8,-1);
  \node[font=\small] at (4,-10.5)
    {$V_{1,n}=0,\quad V_{2,n}\not\to 0$\\\textit{power loss}};
\end{scope}

\begin{scope}[shift={(23,0)}]
  \node[font=\small\bfseries] at (4,2.9) {Block randomization};
  \node[font=\scriptsize] at (4,2.0) {marginal deviations $a_k-\rho_n$};
  \foreach \k in {1,...,8}{
    \draw[zero] (\k-1,0.85) rectangle ++(1,0.65);
  }
  \node[font=\scriptsize, anchor=east] at (-0.2,1.18) {$V_1$};
  \node[font=\scriptsize] at (4,0.2) {pairwise deviations $b_{k\ell}-\tau_n$};
  \foreach \i in {1,...,8}{
    \foreach \j in {1,...,8}{
      \ifnum\i=\j
        \draw[diag] (\j-1,-\i-1) rectangle ++(1,1);
      \else
        \pgfmathtruncatemacro{\bi}{ifthenelse(\i<5,1,2)}
        \pgfmathtruncatemacro{\bj}{ifthenelse(\j<5,1,2)}
        \ifnum\bi=\bj
          \draw[negweak] (\j-1,-\i-1) rectangle ++(1,1);
        \else
          \draw[posweak] (\j-1,-\i-1) rectangle ++(1,1);
        \fi
      \fi
    }
  }
  \draw[black!45, line width=0.6pt] (0,-9) rectangle (8,-1);
  \draw[black!50, dashed, line width=0.5pt] (4,-9) -- (4,-1);
  \draw[black!50, dashed, line width=0.5pt] (0,-5) -- (8,-5);
  \node[font=\small] at (4,-10.5)
    {$V_{1,n}=0,\quad V_{2,n}=O(n^{-2})$\\\textit{optimal power}};
\end{scope}

\node[font=\scriptsize, anchor=east] at (9.8,-13.0)
  {$b_{k\ell}-\tau_n$\,:};
\draw[negstrong] (10.50,-13.40) rectangle ++(1.3,0.8);  
\draw[negweak]   (11.95,-13.40) rectangle ++(1.3,0.8);  
\draw[zero]      (13.90,-13.40) rectangle ++(1.3,0.8);  
\draw[posweak]   (15.85,-13.40) rectangle ++(1.3,0.8);  
\draw[posstrong] (17.30,-13.40) rectangle ++(1.3,0.8);  
\node[font=\scriptsize] at (11.88,-14.15) {neg.};
\node[font=\scriptsize] at (14.55,-14.15) {$\approx 0$};
\node[font=\scriptsize] at (17.23,-14.15) {pos.};

\end{tikzpicture}
\caption{Each panel shows the marginal deviations $a_k-\rho_n$ (top strip)
and pairwise deviations $b_{k\ell}-\tau_n$ (heatmap).
All three examples have $V_{1,n}=0$ (flat top strips),
but their pairwise structure differs markedly:
uniform randomization is entirely flat ($V_{2,n}=0$),
cyclic shifts exhibit persistent banded dependence ($V_{2,n}\not\to 0$),
and the block design produces only weak block-level deviations
that vanish asymptotically ($V_{2,n}=O(n^{-2})$).
The second-order dispersion $V_{2,n}$ thus governs the
distinction between optimal and non-optimal randomization.}
\label{fig:dispersion-illustration}
\end{figure}
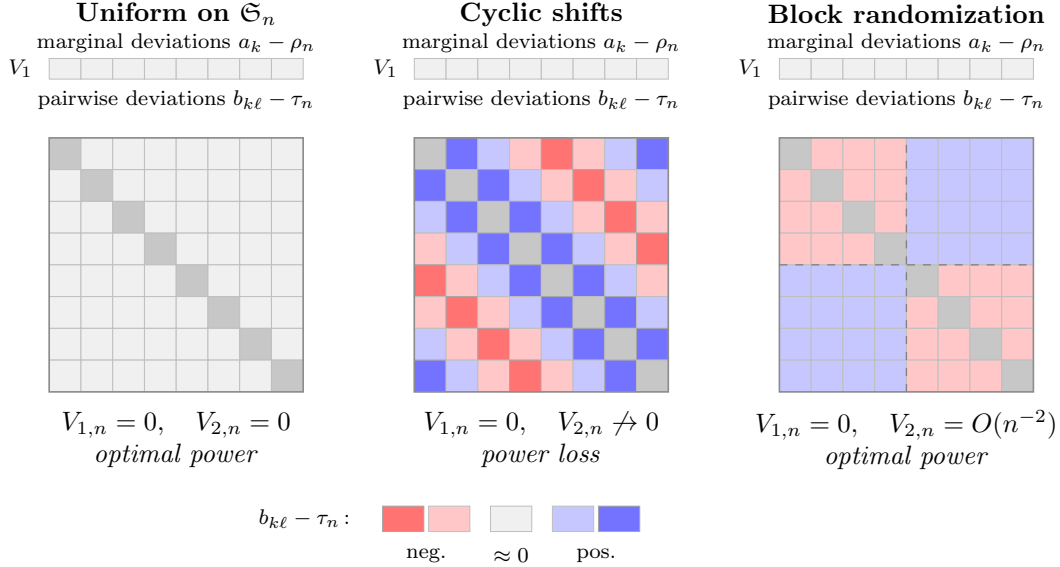

\begin{proposition}[Uniform distribution on $\Sn$]\label{prop:ex-uniform}
Let $q_n=\mathrm{Uniform}(\Sn)$. Then $V_{1,n}=V_{2,n}=0$ for every $n$.
\end{proposition}

\begin{proof}
Under the uniform distribution, $\pi(k)$ is uniform on $[n]$, so $a_k=\rho_n$ for all $k$.
The pair $(\pi(k),\pi(\ell))$ is a uniform sample of size two without replacement from $[n]$,
so $b_{k\ell}=\tau_n$ for all $k\neq\ell$.
Both dispersions vanish identically.
\end{proof}

The following example shows that $V_{1,n}=0$ does not imply $V_{2,n}\to 0$.

\begin{proposition}[Uniform distribution on the cyclic group]\label{prop:ex-cyclic}
Let $q_n=\mathrm{Uniform}(C_n)$ where $C_n$ is the cyclic subgroup of $\Sn$.
Then $V_{1,n}=0$ for every $n$, while
\[
V_{2,n}
\ \rightarrow\
\rho_*^{\,3}\!\left(\frac23-\rho_*\right)>0,
\]
where $\rho_*:=\min(\rho,1-\rho)$.
In particular, Assumption~\ref{ass:V} fails, and by Theorem~\ref{thm:necessary}
there exists at least one nominal level at which the cyclic permutation critical value does not stabilize at the $z$-test benchmark.
\end{proposition}

\begin{proof}
See Appendix~\ref{app:proof-cyclic}.
\end{proof}

The cyclic group provides a clean illustration of why pairwise regularity matters.
Each unit is equally likely to be assigned to treatment ($a_k=\rho_n$),
so the marginal assignment probabilities are perfectly balanced.
However, cyclic shifts move ``neighboring'' indices in lockstep,
so if unit $k$ is treated, the cyclic structure forces $k+1,\dots,k+n_1-1$
to be treated as well.
This rigid dependence is invisible at the marginal level but fully captured by $V_{2,n}$.
For the balanced design $\rho=1/2$, the limiting value of $V_{2,n}$ is $1/48\approx 0.021$.

The next example demonstrates that a restricted permutation group can still satisfy both dispersion conditions,
provided the pairwise deviations are of sufficiently small order.

\begin{proposition}[Block permutation design]\label{prop:ex-block}
Partition $A_n$ into two blocks $B_1,B_2$ of size $n_1/2$
and $A_n^c$ into two blocks $B_3,B_4$ of size $n_0/2$.
Form the matched pairs $\mathcal P_1:=B_1\cup B_3$ and $\mathcal P_2:=B_2\cup B_4$,
each of size $N=n/2$, and let $q_n$ draw independent $\mathrm{Uniform}(\mathfrak S_N)$ permutations
within each pair.
Then $V_{1,n}=0$ and $V_{2,n}=O(n^{-2})\to 0$.
In particular, Assumption~\ref{ass:V} holds and the block permutation test attains optimal power.
\end{proposition}

\begin{proof}
See Appendix~\ref{app:proof-block}.
\end{proof}

The block design uses only $((n/2)!)^2$ out of $n!$ possible permutations, which is a substantial restriction.
Yet because the pairwise deviations from the joint probabilities $\tau_{ij}:=\Pr(w_i=1,w_j=1)$ are small ($O(1/n)$ each),
the aggregate dispersion $V_{2,n}$ still vanishes.
This stands in contrast to the cyclic group,
where the rigid lockstep structure generates $O(1)$ pairwise deviations
and $V_{2,n}\not\to 0$.
The comparison illustrates that the size of the permutation group is not what matters
for optimal power; rather, it is the regularity of the induced assignment probabilities.

\section{When non-uniform randomization dominates}\label{sec:nonuniform}

The results of Section~\ref{sec:optimal} provide a necessary and sufficient characterization
of optimal Pitman local power within the Pitman local model of Assumption~\ref{ass:pitman}.
Under that model, $V_{1,n}\to 0$ and $V_{2,n}\to 0$ are equivalent to the conditional
permutation distribution converging to a standard normal and the test attaining the $z$-test benchmark.
In particular, within the Pitman local model, the dispersions alone determine power,
and any two permutation laws $q_n$ with $V_{1,n}\to 0$ and $V_{2,n}\to 0$ yield identical asymptotic power.

Real data, however, frequently depart from the strict i.i.d.\ noise assumption of the Pitman local model.
A canonical example is a stratified design in which the observations carry between-stratum heterogeneity
on top of a local treatment effect.
In this richer setting, the conclusions of Section~\ref{sec:optimal} no longer apply directly,
and the dispersions, while still informative, cease to determine power on their own.
Although $V_{1,n}\to0$ and $V_{2,n}\to0$ continue to ensure that the conditional permutation distribution
converges to a Gaussian, the scale of that Gaussian depends on how $q_n$
interacts with the data-generating process,
and specifically on how much of the between-stratum variation
the permutation variance absorbs.

This section demonstrates, through a concrete stratified design,
that this scale dependence opens room for strict power gains from non-uniform randomization.
Two permutation laws may both satisfy $V_{1,n}=0$ and $V_{2,n}\to 0$, yet the one that respects
the stratum structure (within-stratum permutation) strictly dominates the one that ignores it
(full permutation on $\Sn$), because its permutation variance reflects only the within-stratum noise
rather than the inflated total variation.
This mirrors the classical variance-reduction benefit of stratification in parametric testing,
and our analysis places it formally within the generalized permutation framework
of Ramdas et al.~\cite{ramdas2023permutation}.

\subsection{Setup}\label{sec:strat-design}

To demonstrate that the dispersion conditions do not fully determine power outside the Pitman local model,
we consider a stratified two-sample design with between-stratum heterogeneity.
Fix $n$ divisible by $4$ and set $n_1=n_0=n/2$.
Partition $[n]$ into two strata of equal size,
\[
\mathcal S_1:=\{1,\dots,n/2\},
\qquad
\mathcal S_2:=\{n/2+1,\dots,n\},
\]
where $\mathrm{s}(i)\in\{1,2\}$ denotes the stratum of unit $i$.
For this stratified design we adopt a reference labeling that places $n/4$ treated units
in each stratum, overriding the convention of Section~\ref{sec:setup}:
\[
w_0 \;:=\; \bigl(\underbrace{1,\dots,1}_{n/4},\underbrace{0,\dots,0}_{n/4},
              \underbrace{1,\dots,1}_{n/4},\underbrace{0,\dots,0}_{n/4}\bigr)^\top,
\]
so that
$A_n := \{i : w_{0,i}=1\} = \{1,\dots,n/4\} \cup \{n/2+1,\dots,3n/4\}$.
This labeling ensures that the treatment assignment is balanced \emph{within} each stratum,
which is essential for the within-stratum permutation strategy below to be non-trivial.
All quantities defined in Section~\ref{sec:gen-perm}
($w_0$, $A_n$, $W(\tau)$, $T_n$) are taken with respect to this labeling throughout
Section~\ref{sec:nonuniform}.

The null and alternative hypotheses are
\begin{align*}
H_0 &:\ X_1,\dots,X_n \iid N(0,1), \\
H_1 &:\ X_i = \mu_{\mathrm{s}(i)} + \frac{\delta}{\sqrt{n}}\,w_{0,i} + \varepsilon_i,
\qquad \varepsilon_i\iid N(0,1),
\end{align*}
where $\mu_1=\mu$, $\mu_2=-\mu$ for a fixed nuisance parameter $\mu>0$,
and $\delta>0$ is a fixed local treatment effect.

The test statistic is the unstandardized difference in means,
\[
T:=\bar X_{\mathrm{trt}}-\bar X_{\mathrm{ctrl}}.
\]
We deliberately use this unstandardized form rather than the oracle-normalized $T_n$ of
Section~\ref{sec:gen-perm}.
The power gap between the two strategies originates in their permutation variances,
and normalizing $T$ by a known $\sigma_\varepsilon$ would obscure that gap.
Both strategies reject for large values of $T$,
compared against the conditional $(1-\alpha)$-quantile of the permutation distribution of $T$.

\medskip\noindent
\textbf{Strategy~1} (full permutation).
Draw $\pi\sim\mathrm{Uniform}(\Sn)$, treating all $n!$ permutations as equally likely.
This ignores the stratum structure, so the permutation variance absorbs both within-
and between-stratum variation.

\medskip\noindent
\textbf{Strategy~2} (within-stratum permutation).
Draw independent uniform permutations within each stratum,
so that $q_n=\mathrm{Uniform}(\mathfrak S_{n/2}\times\mathfrak S_{n/2})$
and no unit ever crosses the stratum boundary.
The permutation variance then reflects only within-stratum noise,
leaving the between-stratum component out of the reference distribution.

\begin{remark}[Role of $\pi_0^{-1}$]\label{rem:pi0-strat}
For Strategy~2, $q_n$ is uniform on the subgroup $G:=\mathfrak S_{n/2}\times\mathfrak S_{n/2}$.
Since the uniform distribution on any finite group is invariant under right multiplication,
$q_n^{(\pi_0)}=q_n$ for every $\pi_0\in G$.
The composition $\pi_m\circ\pi_0^{-1}$ is again uniform on $G$, independently of $\pi_0$.
\end{remark}

Both strategies satisfy $V_{1,n}=0$ and $V_{2,n}\to 0$ (the within-stratum design 
is a special case of the block design of Proposition~\ref{prop:ex-block} 
with strata as blocks), so within the Pitman local model
they would be equivalent. Under $H_0$, the data are i.i.d.\ and hence exchangeable,
so Theorem~\ref{thm:ramdas-valid} guarantees that both yield valid $p$-values.
The alternative $H_1$ departs from the Pitman local model of Assumption~\ref{ass:pitman},
since the location term $\mu_{\mathrm{s}(i)}$ varies across strata rather than being common.
Consequently, Theorems~\ref{thm:cond-normal}--\ref{thm:power} do not apply directly
to the data-generating process of $H_1$, and the analysis below is self-contained.
The precise role of the dispersions in this richer setting is clarified at the end of this section.

\subsection{Permutation variances}\label{sec:perm-var}

Fix the data $X=(X_1,\dots,X_n)$.
Under each strategy, the permutation distribution of $T$ is the law of $T(X,W(\tau))$
where $\tau$ is drawn from the strategy's randomization law and $X$ is held fixed.
We write $\Var_{\mathrm{full}}(T\mid X)$ and $\Var_{\mathrm{strat}}(T\mid X)$
for the conditional variances of $T$ under this permutation distribution,
for Strategies~1 and~2 respectively.
By Remark~\ref{rem:pi0-strat}, these variances do not depend on the base draw~$\pi_0$.

\begin{proposition}[Permutation variance under full randomization]\label{prop:var-full}
Under Strategy~1,
\[
\Var_{\mathrm{full}}(T\mid X)
=\frac{4}{n(n-1)}\sum_{i=1}^n \bigl(X_i-\bar X\bigr)^2.
\]
Under $H_1$,
$n\cdot\Var_{\mathrm{full}}(T\mid X)\toP 4(1+\mu^2)$.
\end{proposition}

\begin{proof}
See Appendix~\ref{app:proof-var-full}.
\end{proof}

\begin{proposition}[Permutation variance under within-stratum randomization]\label{prop:var-strat}
Under Strategy~2,
\[
\Var_{\mathrm{strat}}(T\mid X)
=\frac{Q_1+Q_2}{(n/2)(n/2-1)},
\]
where $Q_s:=\sum_{i\in\mathcal S_s}(X_i-\bar X_s)^2$.
Under $H_1$,
$n\cdot\Var_{\mathrm{strat}}(T\mid X)\toP 4$.
\end{proposition}

\begin{proof}
See Appendix~\ref{app:proof-var-strat}.
\end{proof}

\subsection{Power comparison}\label{sec:power-compare}

Let $c_{\mathrm{full}}(\alpha\mid X)$ and $c_{\mathrm{strat}}(\alpha\mid X)$
denote the $(1-\alpha)$-quantiles of the conditional permutation distributions of $T$
under Strategies~1 and~2, respectively.
The corresponding one-sided tests reject $H_0$ when
$T>c_{\mathrm{full}}(\alpha\mid X)$ or $T>c_{\mathrm{strat}}(\alpha\mid X)$,
and are valid in finite samples by Theorem~\ref{thm:ramdas-valid}.
The power of each test under $H_1$ is denoted
$\mathrm{Power}_{\mathrm{full}}:=\mP(T>c_{\mathrm{full}}(\alpha\mid X))$
and $\mathrm{Power}_{\mathrm{strat}}:=\mP(T>c_{\mathrm{strat}}(\alpha\mid X))$.

\begin{proposition}[Non-uniform dominates uniform]\label{prop:power-compare}
Under $H_1$ and a one-sided test at level $\alpha\in(0,1/2)$,
the asymptotic powers are
\begin{align}
\mathrm{Power}_{\mathrm{full}}
&= 1-\Phi\!\left(z_{1-\alpha}\sqrt{1+\mu^2}-\frac{\delta}{2}\right), \label{eq:power-full}\\
\mathrm{Power}_{\mathrm{strat}}
&= 1-\Phi\!\left(z_{1-\alpha}-\frac{\delta}{2}\right). \label{eq:power-strat}
\end{align}
For any $\mu>0$,
$\mathrm{Power}_{\mathrm{strat}} > \mathrm{Power}_{\mathrm{full}}$.
Moreover, $\mathrm{Power}_{\mathrm{full}}\to 0$ as $\mu\to\infty$ while $\mathrm{Power}_{\mathrm{strat}}$ is independent of $\mu$.
\end{proposition}

\begin{proof}
See Appendix~\ref{app:proof-power-compare}.
\end{proof}

The proof combines three ingredients.
First, by balanced design, the stratum effects $\mu_{\mathrm{s}(i)}$ cancel in $T$ under $H_1$,
so that $\sqrt{n}\,T/2\dto \N(\delta/2,1)$.
Second, a H\'ajek-type permutation CLT (Lemma~\ref{lem:perm-CLT-strat} in the Appendix)
shows that, under each strategy, the standardized permutation distribution of $T$ converges to $\N(0,1)$.
It follows that the critical value satisfies
$c(\alpha\mid X)=z_{1-\alpha}\,\hat\sigma\,(1+o_P(1))$,
where $\hat\sigma^2$ is the permutation variance.
Third, Propositions~\ref{prop:var-full}--\ref{prop:var-strat} identify the limiting variances
at $4(1+\mu^2)/n$ and $4/n$ respectively.
Combining these via Slutsky's theorem yields
\eqref{eq:power-full}--\eqref{eq:power-strat}.
The strict dominance $\mathrm{Power}_{\mathrm{strat}}>\mathrm{Power}_{\mathrm{full}}$
for $\mu>0$ then follows because $\sqrt{1+\mu^2}>1$ inflates the full permutation's critical value
without any compensating gain.
Table~\ref{tab:power} reports the resulting asymptotic powers for representative values of $\mu$ and $\delta$.

\begin{table}[H]
\centering
\caption{Asymptotic power at level $\alpha=0.05$, computed from the
closed-form expressions \eqref{eq:power-full}--\eqref{eq:power-strat}.
Both tests are valid under $H_0$ (Theorem~\ref{thm:ramdas-valid}).
Finite-sample simulation results are reported in
Section~\ref{sec:sim-stratified}.}
\label{tab:power}
\smallskip
\scalebox{0.9}{%
\setlength{\tabcolsep}{8pt}%
\begin{tabular}{@{}
  c
  S[table-format=1.1]
  S[table-format=1.3]
  S[table-format=1.3]
  S[table-format=1.3]
  @{}}
\toprule
& {$\mu$} & \multicolumn{3}{c}{Asymptotic power} \\
\cmidrule(lr){3-5}
& & {Full permutation} & {Within-stratum} & {Gain} \\
\midrule
\multirow{5}{*}{$\delta = 2$}
  & 0   & 0.260 & 0.260 & 0.000 \\
  & 0.5 & 0.201 & 0.260 & 0.059 \\
  & 1   & 0.092 & 0.260 & 0.167 \\
  & 2   & 0.004 & 0.260 & 0.256 \\
  & 3   & {$<0.001$} & 0.260 & 0.259 \\
\midrule
\multirow{4}{*}{$\delta = 3$}
  & 0   & 0.442 & 0.442 & 0.000 \\
  & 1   & 0.204 & 0.442 & 0.238 \\
  & 2   & 0.015 & 0.442 & 0.428 \\
  & 3   & {$<0.001$} & 0.442 & 0.442 \\
\bottomrule
\end{tabular}%
}
\end{table}

The power difference arises because the full permutation mixes units across strata,
so its permutation variance absorbs the between-stratum variation $\mu^2$,
inflating it by a factor of $1+\mu^2$.
The within-stratum permutation confines resampling to within each stratum,
so its permutation variance captures only within-stratum noise.

As noted at the start of this section, both strategies satisfy 
$V_{1,n}=0$ and $V_{2,n}\to 0$,
so the dispersion framework of Section~\ref{sec:optimal} cannot distinguish between them.
Proposition~\ref{prop:power-compare} confirms that outside the Pitman local model,
the scale of the permutation distribution becomes a separate determinant of power
beyond the dispersions, driven here by how $q_n$ interacts with the between-stratum variation.

\section{Simulations}\label{sec:simulations}

We present Monte Carlo simulations that confirm the theoretical results
of Sections~\ref{sec:optimal} and~\ref{sec:nonuniform}.
Section~\ref{sec:sim-dispersion} examines how different permutation distributions
affect power under the Pitman local model,
confirming the role of the dispersions $(V_{1,n},V_{2,n})$.
Section~\ref{sec:sim-stratified} demonstrates the power advantage of within-stratum
permutation over full permutation in the presence of nuisance stratification.

\paragraph{Common simulation parameters.}
Throughout both subsections we fix the local effect size $\delta=3$,
nominal level $\alpha=0.05$, balanced design $\rho=n_1/n=1/2$,
and i.i.d.\ Gaussian noise $\varepsilon_i\iid\N(0,1)$.
Each test uses $B=199$ Monte Carlo permutations, and we run
$N_{\mathrm{rep}}=20{,}000$ independent replications per configuration
to estimate rejection rates.

For each replication, data $X$ are generated under the relevant model,
the observed test statistic $T^{\mathrm{obs}}$ is computed,
$B$ permutations $\pi_1,\dots,\pi_B$ are drawn from $q_n$,
and the resampled statistics $T_b := T(X, W(\pi_b\circ\pi_0^{-1}))$
are computed for $b=1,\dots,B$.
For both strategies considered here, $q_n$ is uniform on a group,
so right-translation preserves the law and $\pi_0$ plays no role in practice.
The resulting one-sided $p$-value $P_n(X)$ of Theorem~\ref{thm:ramdas-valid} (with $M=B$)
is compared against $\alpha$.
The empirical rejection rate over the $N_{\mathrm{rep}}$ replications
estimates the power under $H_1$; under $H_0$, the same scheme estimates the type~I error.
With $N_{\mathrm{rep}}=20{,}000$, the Monte Carlo standard error of each
rejection-rate estimate is below $0.0036$.
We use the oracle-normalized $T_n$ throughout Section~\ref{sec:sim-dispersion}
(so $T^{\mathrm{obs}}=T_n^{\mathrm{obs}}$), to align with the theory of Section~\ref{sec:optimal};
Theorem~\ref{thm:AL} ensures that studentized variants yield the same asymptotic power when the variance estimator is consistent along $q_n^{(\pi_0)}$-typical assignments.
In Section~\ref{sec:sim-stratified}, $T$ denotes the unstandardized difference in means.

The theoretical $z$-test power benchmark at these parameters is
\[
1-\Phi\Bigl(z_{0.95}-\delta\sqrt{\rho(1-\rho)}/\sigma_\varepsilon\Bigr)
=1-\Phi(1.645-1.5)\approx 0.442.
\]

\subsection{Power under different permutation distributions}\label{sec:sim-dispersion}

The central question of Section~\ref{sec:optimal} is whether,
\emph{given an arbitrary permutation distribution $q_n$,
the resulting test achieves optimal Pitman local power.}
The dispersions $(V_{1,n},V_{2,n})$ serve as the diagnostic tool
that answers this question.
In this subsection we confirm empirically that the dispersions correctly separate
optimal from non-optimal permutation distributions.

\paragraph{Setup.}
We generate data under the Pitman local model of Assumption~\ref{ass:pitman}:
$X_i=\frac{\delta}{\sqrt{n}}w_i+\varepsilon_i$
with $\delta=3$ and $\varepsilon_i\iid\N(0,1)$.
The observed assignment $w_0$ places the first $n_1=n/2$ units in treatment.
The test statistic is the standardized difference in means~\eqref{eq:Tn}.
We vary the sample size over $n\in\{20,40,60,100,160,200,300,500,800,1000,1500,2000\}$.

\paragraph{Permutation distributions compared.}
Beyond the three canonical examples of Section~\ref{sec:examples-dispersion}
(Uniform, Block, Cyclic), which cover the cases $V_{1,n}=0$ with $V_{2,n}\to 0$
and $V_{1,n}=0$ with $V_{2,n}\not\to 0$,
we include two additional distributions (Small-support, Noisy-sort)
that illustrate the distinct effect of non-vanishing $V_{1,n}$.
The resulting five distributions jointly cover the three qualitatively different
dispersion regimes: $V_{1,n}=V_{2,n}=0$, $V_{1,n}=0$ with $V_{2,n}\not\to 0$,
and $V_{1,n}>0$:
\begin{enumerate}[label=(\roman*)]
\item \emph{Uniform} on $\Sn$:
treatment indices are a simple random sample of size $n_1$ from $[n]$.
Dispersions: $V_{1,n}=V_{2,n}=0$ (Proposition~\ref{prop:ex-uniform}).
Prediction: optimal power.

\item \emph{Block permutation}:
$[n]$ is partitioned into two matched pairs of size $N=n/2$,
and independent uniform permutations are applied within each pair
(Proposition~\ref{prop:ex-block}).
Dispersions: $V_{1,n}=0$, $V_{2,n}=O(n^{-2})\to 0$.
Prediction: optimal power.

\item \emph{Cyclic group}:
$\pi=c^J$ where $c$ is the cyclic shift and $J\sim\mathrm{Uniform}\{0,\dots,n-1\}$.
Dispersions: $V_{1,n}=0$, $V_{2,n}\to\rho_*^3(2/3-\rho_*)>0$
(Proposition~\ref{prop:ex-cyclic}).
Prediction: non-optimal (pairwise condition fails despite marginal balance).

\item \emph{Small-support} ($K=20$):
$20$ permutations are drawn uniformly from $\Sn$ once before the experiment begins,
and each Monte Carlo draw selects one of these $20$ permutations uniformly at random.
Dispersions: $V_{1,n}>0$, $V_{2,n}>0$.
Prediction: non-optimal (both conditions fail).

\item \emph{Noisy-sort} ($\eta=2.0$):
the permutation $\pi$ is defined by sorting units by the perturbed scores
$i+\eta Z_i$ where $Z_i\iid\N(0,1)$.
Units near each other in the original ordering tend to stay close.
Dispersions: $V_{1,n}>0$, full support on $\Sn$.
Prediction: non-optimal (marginal condition fails).
\end{enumerate}

\paragraph{Implementation.}
For strategies~(i)--(iii), the permutation distribution is uniform on a subgroup of $\Sn$,
so the right-translate $\pi_b\circ\pi_0^{-1}$ is again uniformly distributed
on the same subgroup (cf.\ Remark~\ref{rem:pi0-strat}).
We may therefore draw fresh assignments directly without explicitly computing $\pi_0$.
For strategies~(iv)--(v), we implement the full $\pi_0$-based protocol
of Theorem~\ref{thm:ramdas-valid}:
draw $\pi_0,\pi_1,\dots,\pi_B\iid q_n$ and form the resampled assignments
$W(\pi_b\circ\pi_0^{-1})$.

\begin{figure}[t]
\centering
\includegraphics[width=\textwidth]{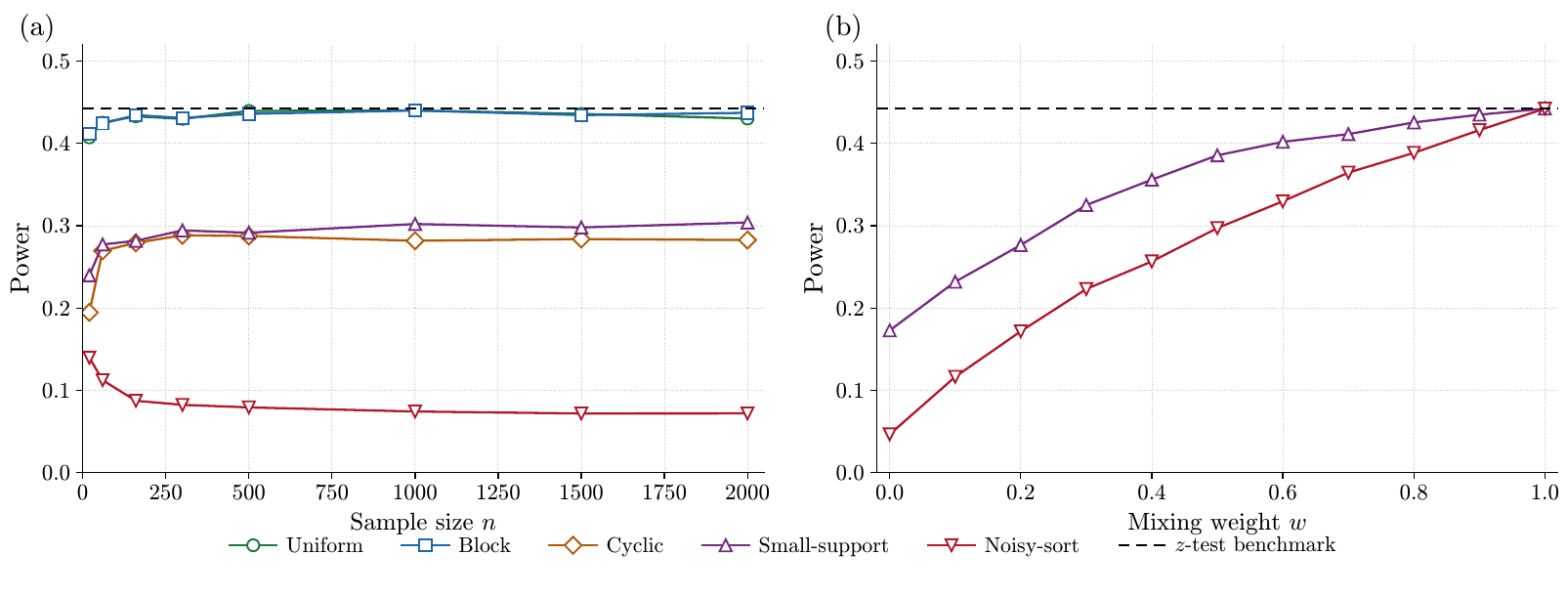}
\caption{\textbf{Experiment 1: dispersion regimes.}
\textbf{(a)} Power vs.\ sample size for five permutation strategies,
illustrating the dispersion-based power predictions of
Theorems~\ref{thm:cond-normal}--\ref{thm:necessary}.
Uniform and Block both have $V_{1,n}=0$ and $V_{2,n}\to 0$,
so their power approaches the $z$-test benchmark ($0.442$).
Cyclic has $V_{2,n}\not\to 0$ and stabilizes below the benchmark,
while Small-support and Noisy-sort have $V_{1,n}>0$ and suffer further power loss.
\textbf{(b)} Mixture power for
$q_n^{(w)} = w\cdot\mathrm{Unif}(\mathfrak S_n) + (1-w)\cdot Q$ at $n=200$
as $w$ varies from $0$ to $1$,
with $Q$ either Small-support ($K=20$) or Noisy-sort ($\eta=2.0$).
Both mixtures rise monotonically toward the benchmark,
confirming that the optimality condition of Theorem~\ref{thm:cond-normal}
is robust to mixture contamination.
Parameters: $\delta=3$, $\alpha=0.05$, $B=199$, $20{,}000$ replications.}
\label{fig:exp1}
\end{figure}

\paragraph{Results.}
Figure~\ref{fig:exp1} reports the empirical rejection rate as a function of $n$,
and several patterns emerge, all consistent with the theory.
First, the \emph{Uniform} and \emph{Block} tests (both satisfying $V_{1,n}\to 0$ and $V_{2,n}\to 0$)
converge to the $z$-test benchmark $0.442$ as $n$ grows, confirming Theorem~\ref{thm:power};
the two curves are nearly indistinguishable for $n\ge 100$.
Second, the \emph{Cyclic} test, which has $V_{1,n}=0$ but $V_{2,n}\to 1/48\approx 0.021$,
converges to a power of approximately $0.28$, well below the benchmark,
confirming that marginal balance alone is insufficient and that $V_{2,n}\to 0$
is genuinely necessary (Theorem~\ref{thm:necessary});
the pairwise irregularity $V_{2,n}\to 1/48$ translates into an absolute power loss of roughly $0.16$,
offering a concrete calibration for how much pairwise dependence the dispersion framework ``permits''.
Third, the \emph{Small-support} test ($V_{1,n}>0$, $V_{2,n}>0$) stabilizes at approximately $0.30$
and the \emph{Noisy-sort} test ($V_{1,n}>0$, full support) achieves the lowest power at $\approx 0.08$,
illustrating that $V_{1,n}$ has a more severe impact on power than $V_{2,n}$ alone.
Finally, separate simulations (not shown) confirm that all five tests remain valid under $H_0$ ($\delta=0$),
with empirical rejection rates at or below $\alpha=0.05$ for every strategy and every $n$,
consistent with Theorem~\ref{thm:ramdas-valid}.

\paragraph{Mixture contamination.}
To examine the sensitivity of the dispersion framework to the purity of $q_n$,
we consider convex mixtures
\[
q_n^{(w)} = w\cdot\mathrm{Unif}(\Sn) + (1-w)\cdot Q,
\qquad w\in[0,1],
\]
where $Q$ is either Small-support ($K=20$) or Noisy-sort ($\eta=2.0$).
At $w=0$, the randomization is pure $Q$ and suffers the dispersion penalty seen in
panel~(a); at $w=1$, it is uniform and attains the $z$-test benchmark $0.442$.
Panel~(b) of Figure~\ref{fig:exp1} plots both mixtures at $n=200$.
In both cases the power rises smoothly and monotonically with $w$,
and the benchmark is approximately attained once $w\ge 0.8$.
The Noisy-sort mixture recovers more slowly than the Small-support mixture,
consistent with its greater dispersion penalty in panel~(a).
At $w=0.2$, the Small-support mixture already reaches $\approx 0.27$
while the Noisy-sort mixture is at $\approx 0.17$,
and both converge to $\approx 0.44$ at $w=1$.
This behavior is consistent with the dispersion formula
$V_{j,n}(q_n^{(w)}) = (1-w)^2\,V_{j,n}(Q)$ for $j=1,2$,
since the uniform component contributes zero to both dispersions.
Any fixed $w<1$ leaves residual dispersion of order $(1-w)^2$,
which decreases continuously to zero as $w\to 1$.
The experiment thus confirms that the optimality condition of
Theorem~\ref{thm:cond-normal} is robust to mixture contamination,
in a manner consistent with the $(1-w)^2$ dispersion formula.


\subsection{Power comparison in the stratified design}\label{sec:sim-stratified}

\paragraph{Setup.}
We use the design, hypotheses, and test statistic of Section~\ref{sec:strat-design}:
two strata of equal size $n/2$ with stratum means $\mu_1=\mu$ and $\mu_2=-\mu$,
$n/4$ units treated in each stratum, data generated as
$X_i=\mu_{\mathrm{s}(i)}+\frac{\delta}{\sqrt{n}}w_{0,i}+\varepsilon_i$
with $\varepsilon_i\iid\N(0,1)$, and test statistic $T=\bar X_{\mathrm{trt}}-\bar X_{\mathrm{ctrl}}$.
We compare the two strategies of Section~\ref{sec:strat-design}
(full permutation and within-stratum permutation),
both valid under $H_0$ by Theorem~\ref{thm:ramdas-valid} and satisfying 
$V_{1,n}=0$ and $V_{2,n}\to 0$.
The theoretical asymptotic powers from Proposition~\ref{prop:power-compare} are
\begin{align*}
\mathrm{Power}_{\mathrm{full}}&=1-\Phi\bigl(z_{0.95}\sqrt{1+\mu^2}-\delta/2\bigr),\\
\mathrm{Power}_{\mathrm{strat}}&=1-\Phi\bigl(z_{0.95}-\delta/2\bigr)=0.442.
\end{align*}

\paragraph{Varying the nuisance parameter $\mu$ (fixed $n$).}
We fix $n=200$ and vary $\mu\in\{0,0.5,\ldots,2.5,3\}$.
Figure~\ref{fig:exp2}(a) plots the results.
At $\mu=0$, both strategies achieve power $\approx 0.43$,
close to the common asymptotic limit $0.442$.
As $\mu$ increases, the full permutation test's power drops sharply,
from $0.43$ at $\mu=0$ to $0.37$ at $\mu=0.5$,
$0.21$ at $\mu=1$, and effectively zero by $\mu=2.5$.
The within-stratum test maintains power $\approx 0.43$ across all values of $\mu$.
The simulated rejection rates closely match the theoretical dashed curves.

\paragraph{Varying the sample size $n$ (fixed $\mu$).}
We fix $\mu=1$ and vary $n\in\{40,60,\ldots,1600,2000\}$.
Figure~\ref{fig:exp2}(b) reports the empirical rejection rate of each strategy
together with the asymptotic limits from Proposition~\ref{prop:power-compare},
$\mathrm{Power}_{\mathrm{full}}\approx 0.204$ and $\mathrm{Power}_{\mathrm{strat}}\approx 0.442$.
Two features stand out.
First, both curves are essentially flat in $n$.
The within-stratum test achieves power $\approx 0.43$ already at $n=40$
and tracks its asymptote to within Monte Carlo error thereafter,
while the full permutation test is pinned near $0.20$ across the entire range.
Second, and more importantly, the gap between the two strategies
does not close as $n$ grows; the gap is already fully present at $n=40$.
This confirms that the power loss of full permutation under nuisance heterogeneity
is a genuine asymptotic phenomenon (Proposition~\ref{prop:power-compare}) rather than
a finite-sample artifact.

\begin{figure}[t]
\centering
\includegraphics[width=\textwidth]{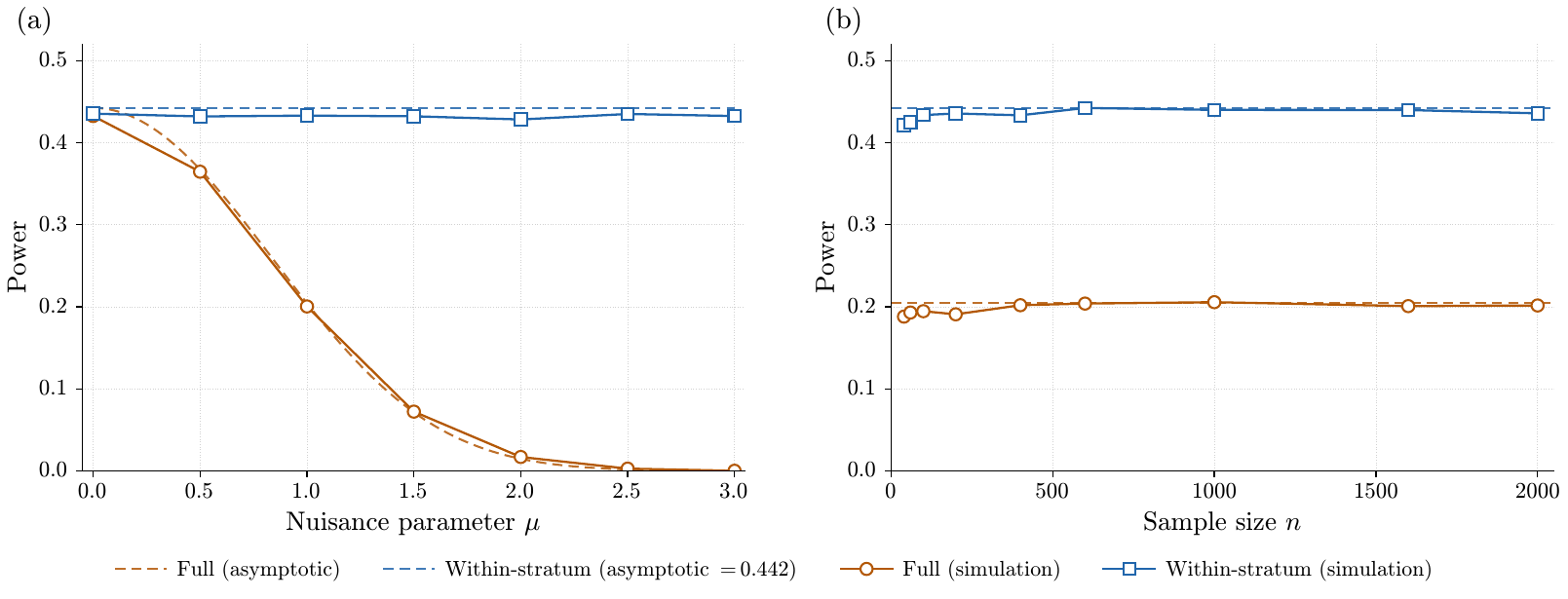}
\caption{\textbf{Experiment 2: stratified design.}
\textbf{(a)} Power as a function of the nuisance parameter $\mu$ at fixed $n=200$.
The full-permutation test loses power rapidly as $|\mu|$ grows,
while within-stratum permutation remains at the asymptotic optimum $0.442$,
confirming Proposition~\ref{prop:power-compare}.
\textbf{(b)} Power as a function of $n$ at fixed $\mu=1$.
Both curves track their asymptotic limits ($0.204$ for full, $0.442$ for within-stratum) 
from moderate $n$ onward.
Parameters: $\delta=3$, $\alpha=0.05$, $B=199$, $20{,}000$ replications.}
\label{fig:exp2}
\end{figure}

\paragraph{Summary.}
The simulation results consistently support the theoretical analysis.
In Section~\ref{sec:sim-dispersion},
the dispersions $(V_{1,n},V_{2,n})$ correctly predict which permutation distributions
achieve optimal power and which do not.
In Section~\ref{sec:sim-stratified},
the asymptotic power formulas of Proposition~\ref{prop:power-compare}
provide an accurate approximation to the finite-sample rejection rates
across the full range of nuisance parameters $\mu$ (panel a of Figure~\ref{fig:exp2})
and sample sizes $n$ (panel b).
The strict dominance of within-stratum permutation over full permutation
is confirmed in every configuration tested, and the gap
$\mathrm{Power}_{\mathrm{strat}}-\mathrm{Power}_{\mathrm{full}}$
is essentially independent of $n$, as predicted by theory.


\section{Conclusion}\label{sec:conclusion}

We have studied the power of generalized permutation tests
when the randomization distribution $q_n$ on $\Sn$ is not necessarily uniform.
Building on the finite-sample validity framework of Ramdas et al.~\cite{ramdas2023permutation},
we introduced two scalar quantities, the first-order dispersion $V_{1,n}$
and the second-order dispersion $V_{2,n}$, and showed that,
within the Pitman local model,
they provide a necessary and sufficient characterization of when
a generalized permutation test attains optimal local power
for the difference-in-means statistic.
Vanishing of both dispersions is sufficient for the conditional permutation distribution
to converge to the standard normal and for the test to match the $z$-test power benchmark
(Theorem~\ref{thm:power});
under Gaussian noise, non-vanishing dispersions cause the permutation distribution
to fail to self-average and the critical values not to stabilize
(Theorem~\ref{thm:necessary}).
A key byproduct of our analysis is the extension of this characterization
to statistics admitting a permutation-typical asymptotic linearization (Theorem~\ref{thm:AL}),
with the standard studentized difference in means as the main example.
Monte Carlo experiments (Section~\ref{sec:simulations})
confirm that the dispersions correctly separate optimal from non-optimal
randomization laws in finite samples.

We have also shown that this characterization is confined to the Pitman local model.
Once the data carry additional structure (for instance, between-stratum heterogeneity
in a stratified design), the scale of the permutation distribution
becomes a separate determinant of power beyond the dispersions.
In such settings, non-uniform randomization laws that respect the nuisance structure
can strictly dominate the uniform distribution even when both satisfy $V_{1,n}=0$ and $V_{2,n}\to0$
(Proposition~\ref{prop:power-compare}).
This formally recovers in the generalized permutation framework
the classical variance-reduction benefit of stratification.

We highlight three directions for further work.
First, the rate at which $V_{1,n}$ and $V_{2,n}$ vanish governs
finite-sample approximations to the permutation distribution.
A Berry--Esseen-type analysis quantifying these rates,
possibly refined to account for the empirical observation that $V_{1,n}$
appears more consequential for power than $V_{2,n}$ (Section~\ref{sec:sim-dispersion}),
would be of practical interest.
Second, while Section~\ref{sec:nonuniform} demonstrates that exploiting
known nuisance structure through the choice of $q_n$ yields power gains,
adaptively constructing $q_n$ when the nuisance structure is only estimated
from the data---without sacrificing finite-sample validity---remains an open problem.
Third, a natural extension of the dispersion framework would adapt it to
non-location alternatives, such as scale alternatives or distributional testing,
where the Pitman local structure takes a different form and the relevant measures
of ``deviation from complete randomization'' may need to be redefined.

\bibliographystyle{plainnat}
\bibliography{reference}

\appendix

\section{Auxiliary lemmas}\label{app:auxiliary}

This section collects four technical lemmas used repeatedly in the main proofs.

\begin{lemma}[Right-translates only relabel $(a_k,b_{k\ell})$]\label{lem:translate}
Fix $s\in\Sn$ and let $\tau=\pi\circ s^{-1}$ with $\pi\sim q_n$,
so that $\tau\sim q_n^{(s)}$.
For all $i\in[n]$ and all pairs $i\neq j$,
\[
\mP(\tau(i)\in A_n)=a_{s^{-1}(i)},\qquad
\mP(\tau(i),\tau(j)\in A_n)=b_{s^{-1}(i),\,s^{-1}(j)}.
\]
Consequently, the dispersions $(V_{1,n},V_{2,n})$ computed under $q_n^{(s)}$
coincide with those computed under $q_n$, for every $s\in\Sn$.
\end{lemma}

\begin{proof}
Since $\tau(i)=\pi(s^{-1}(i))$, the events $\{\tau(i)\in A_n\}$ and 
$\{\tau(i),\tau(j)\in A_n\}$ reduce under $\pi\sim q_n$ to 
$\{\pi(s^{-1}(i))\in A_n\}$ and $\{\pi(s^{-1}(i)),\pi(s^{-1}(j))\in A_n\}$,
yielding $a_{s^{-1}(i)}$ and $b_{s^{-1}(i),s^{-1}(j)}$ respectively.

For the dispersion invariance, write 
$a_i^{(s)}:=\mP(\tau(i)\in A_n)=a_{s^{-1}(i)}$
and $b_{ij}^{(s)}:=b_{s^{-1}(i),s^{-1}(j)}$.
Since $s^{-1}:[n]\to[n]$ is a bijection, the substitution $k=s^{-1}(i)$ gives
\[
\frac1n\sum_{i=1}^n\bigl(a_i^{(s)}-\rho_n\bigr)^2
=\frac1n\sum_{k=1}^n\bigl(a_k-\rho_n\bigr)^2=V_{1,n},
\]
and, treating $(i,j)$ as a pair via $k=s^{-1}(i),\,\ell=s^{-1}(j)$,
$\frac1{n^2}\sum_{i\neq j}(b_{ij}^{(s)}-\tau_n)^2=V_{2,n}$.
\end{proof}

\begin{lemma}[A variance bound for overlap]\label{lem:overlap-bound}
Let $(W,W')$ be $\{0,1\}^n$-valued random vectors that are conditionally i.i.d.\ given $\pi_0$,
and define $O_n:=\frac1n\sum_{i=1}^n W_iW_i'$.
Write $p_i:=\mP(W_i=1\mid\pi_0)$ and 
$p_{ij}:=\mP(W_i=W_j=1\mid\pi_0)$ for $i\neq j$. Then
\[
\Var(O_n\mid \pi_0)
\le \frac{1}{4n}
+\frac{4}{n^2}\sum_{1\le i<j\le n}\bigl|p_{ij}-p_ip_j\bigr|.
\]
\end{lemma}

\begin{proof}
The variance decomposes as
\[
\Var(O_n\mid\pi_0)
=\frac{1}{n^2}\sum_{i=1}^n \Var(W_iW_i'\mid\pi_0)
+\frac{2}{n^2}\sum_{1\le i<j\le n}\Cov(W_iW_i',W_jW_j'\mid\pi_0).
\]

For the diagonal terms, $W_iW_i'\in\{0,1\}$ is Bernoulli, so 
$\Var(W_iW_i'\mid\pi_0)\le 1/4$, and the diagonal sum is at most $1/(4n)$.

For the cross terms, conditional independence of $W$ and $W'$ given $\pi_0$ gives
$\E[W_iW_i'W_jW_j'\mid\pi_0]=\E[W_iW_j\mid\pi_0]\cdot\E[W_i'W_j'\mid\pi_0]=p_{ij}^2$,
while $\E[W_iW_i'\mid\pi_0]\cdot\E[W_jW_j'\mid\pi_0]=p_i^2 p_j^2$.
Therefore
\[
\Cov(W_iW_i',W_jW_j'\mid\pi_0)
=p_{ij}^2-p_i^2p_j^2
=(p_{ij}-p_ip_j)(p_{ij}+p_ip_j).
\]
Since $0\le p_{ij},p_i,p_j\le 1$, we have $p_{ij}+p_ip_j\le 2$,
hence $|\Cov(W_iW_i',W_jW_j'\mid\pi_0)|\le 2|p_{ij}-p_ip_j|$.
Substituting into the variance expansion yields the claimed inequality.
\end{proof}

\begin{lemma}[Two-randomization criterion]\label{lem:two-rand}
Let $W$ be a random vector such that $W$ is independent of $X$ given $\pi_0$,
and let $W'$ be an independent copy of $W$ given $\pi_0$, also independent of $X$.
Set $T_n:=T_n(X,W)$ and $T_n':=T_n(X,W')$, so that $(T_n,T_n')$ is i.i.d.\
conditional on $(X,\pi_0)$.
Assume that the conditional law of $(T_n,T_n')$ given $\pi_0$
converges weakly in probability to the law of $(Z,Z')$.
Equivalently, for every bounded Lipschitz function $g:\mathbb R^2\to\mathbb R$,
\[
\E\{g(T_n,T_n')\mid \pi_0\}
\toP
\E\{g(Z,Z')\},
\]
where $Z,Z'$ are i.i.d.\ with continuous cdf $F$.
Then
\[
\sup_{t\in\R}\bigl|\mP(T_n\le t\mid X,\pi_0)-F(t)\bigr|\toP 0.
\]
\end{lemma}

\begin{proof}
Define $p_n(t):=\mP(T_n\le t\mid X,\pi_0)$.
The proof proceeds in two stages: pointwise convergence, then upgrading to uniform.

\medskip
\noindent\emph{Stage 1: pointwise convergence.}
The conditional i.i.d.\ structure given $(X,\pi_0)$ gives
\[
\Var\bigl(p_n(t)\bigr)
=\mP(T_n\le t,\,T_n'\le t)-\mP(T_n\le t)^2,
\]
so we aim to show both the joint and marginal probabilities stabilize at $F(t)^2$ and $F(t)$.
We first upgrade the $\pi_0$-conditional convergence to unconditional convergence.
Let $g:\R^2\to\R$ be bounded continuous and set $Y_n:=\E[g(T_n,T_n')\mid\pi_0]$.
By assumption $Y_n\toP\E[g(Z,Z')]$, and since $|Y_n|\le\|g\|_\infty$,
uniform integrability yields $\E[Y_n]\to\E[g(Z,Z')]$,
i.e., $\E[g(T_n,T_n')]\to\E[g(Z,Z')]$.
In particular, $(T_n,T_n')\dto(Z,Z')$ unconditionally, and $T_n\dto Z$.

Since $F$ is continuous, $(-\infty,t]^2$ is a continuity set for $(Z,Z')$, so
\[
\mP(T_n\le t,\,T_n'\le t)\to F(t)^2
\quad\text{and}\quad
\mP(T_n\le t)\to F(t).
\]
Thus $\Var(p_n(t))\to 0$, and since $\E[p_n(t)]=\mP(T_n\le t)\to F(t)$,
we obtain $p_n(t)\toP F(t)$ for each fixed $t$.

\medskip
\noindent\emph{Stage 2: uniform convergence.}
Fix $\varepsilon>0$. Choose $B>0$ such that $F(-B)\le\varepsilon/4$ and $1-F(B)\le\varepsilon/4$.
By uniform continuity of $F$ on $[-B,B]$, there exists $\eta>0$ with $|F(t)-F(u)|\le\varepsilon/4$ whenever $|t-u|\le\eta$.
Choose a grid $-B=t_0<\cdots<t_J=B$ with mesh at most $\eta$.

From Stage~1 and the union bound, $\max_{0\le j\le J}|p_n(t_j)-F(t_j)|\toP 0$.
For any $t\in[t_j,t_{j+1}]$, monotonicity of both $p_n$ and $F$ gives
\[
|p_n(t)-F(t)|\le \max\{|p_n(t_j)-F(t_j)|,\,|p_n(t_{j+1})-F(t_{j+1})|\}+\varepsilon/4,
\]
so $\sup_{t\in[-B,B]}|p_n(t)-F(t)|\le\max_j|p_n(t_j)-F(t_j)|+\varepsilon/4\toP\varepsilon/4$.
For the tails, monotonicity of $p_n$ gives $p_n(t)\le p_n(-B)$ for $t<-B$, 
so $|p_n(t)-F(t)|\le p_n(-B)+F(-B)\le\varepsilon/2$ with probability tending to one
(using Stage~1 at $t=-B$ and $F(-B)\le\varepsilon/4$);
similarly for $t>B$.
Combining gives $\sup_{t\in\R}|p_n(t)-F(t)|\toP 0$.
\end{proof}

\begin{lemma}[Overlap fluctuations]\label{lem:overlap-V}
Let $\pi,\pi'\iid q_n$ be independent and define
\[
O_n:=\frac1n\sum_{k=1}^n \ind\{\pi(k)\in A_n\}\ind\{\pi'(k)\in A_n\}.
\]
Then
\[
\E[(O_n-\rho_n^2)^2] \ge V_{1,n}^2,
\]
and, moreover, if $V_{1,n}\to 0$, then
\[
\E[(O_n-\rho_n^2)^2]=V_{2,n}+o(1).
\]
\end{lemma}

\begin{proof}
Let $I_k:=\ind\{\pi(k)\in A_n\}$ and $I_k':=\ind\{\pi'(k)\in A_n\}$.
By independence of $\pi$ and $\pi'$, $\E[I_kI_k']=a_k^2$, so
$\E[O_n]=\frac1n\sum_k a_k^2$.
Writing $a_k=\rho_n+d_k$ with $\sum_k d_k=0$ (since $\sum_k a_k=n_1=n\rho_n$),
\[
\E[O_n]-\rho_n^2=\frac1n\sum_k d_k^2=V_{1,n}.
\]
Jensen's inequality gives $\E[(O_n-\rho_n^2)^2]\ge(\E[O_n]-\rho_n^2)^2=V_{1,n}^2$,
proving the first claim.

For the second claim, assume $V_{1,n}\to 0$.
For $k\neq\ell$, independence between $\pi$ and $\pi'$ gives
$\E[I_kI_k'I_\ell I_\ell']=b_{k\ell}^2$, hence
\[
\E[O_n^2]=\frac1{n^2}\sum_k a_k^2+\frac1{n^2}\sum_{k\neq\ell}b_{k\ell}^2.
\]
The combinatorial identity 
$\sum_{k\neq\ell}b_{k\ell}=n_1(n_1-1)$
follows from squaring $\sum_k I_k=n_1$ and taking expectations
(this gives $\sum_k a_k+\sum_{k\neq\ell}b_{k\ell}=n_1^2$, 
and $\sum_k a_k=n_1$ yields the identity).
Consequently, $\sum_{k\neq\ell}(b_{k\ell}-\tau_n)=0$, so
\[
\frac1{n^2}\sum_{k\neq\ell}b_{k\ell}^2
=\frac{n(n-1)}{n^2}\tau_n^2+V_{2,n}.
\]
Expanding $\E[(O_n-\rho_n^2)^2]=\E[O_n^2]-2\rho_n^2\E[O_n]+\rho_n^4$ 
and collecting terms yields $\E[(O_n-\rho_n^2)^2]=V_{2,n}+R_n$, where
\[
R_n=\Big(\frac{n(n-1)}{n^2}\tau_n^2-\rho_n^4\Big)
+\frac1{n^2}\sum_k a_k^2-2\rho_n^2 V_{1,n}.
\]
Each term in $R_n$ is $o(1)$: 
the first is $O(1/n)$ since $\tau_n=\rho_n^2+O(1/n)$;
the second is at most $1/n$ since $0\le a_k\le 1$;
and the third vanishes since $V_{1,n}\to 0$ by assumption.
\end{proof}

\section{Proof of Theorem~\ref{thm:cond-normal}}\label{app:proof-cond-normal}

The proof relies on the following overlap concentration result.

\begin{proposition}[Overlap concentration]\label{prop:overlap}
Conditional on $\pi_0$, let $\tau\sim q_n^{(\pi_0)}$, 
set $W=W(\tau)$ with $W_i=\ind\{\tau(i)\in A_n\}$, 
and let $\tau'$ be an independent copy of $\tau$ given $\pi_0$
with induced $W'=W(\tau')$. Define
\[
O_n:=\frac1n\sum_{i=1}^n W_iW_i',\qquad 
O_n^{(0)}:=\frac1n\sum_{i\in A_n}W_i.
\]
Under Assumptions~\ref{ass:rhon} and~\ref{ass:V},
\[
\E[O_n\mid\pi_0=s]\to\rho^2,\qquad \Var(O_n\mid\pi_0=s)\to 0,
\]
\[
\E[O_n^{(0)}\mid\pi_0=s]\to\rho^2,\qquad \Var(O_n^{(0)}\mid\pi_0=s)\to 0,
\]
with bounds that do not depend on $s\in\Sn$.
Consequently, $O_n\toP\rho^2$ and $O_n^{(0)}\toP\rho^2$.
\end{proposition}

\begin{proof}
Fix $\pi_0=s$.
By Lemma~\ref{lem:translate}, $a_i^{(s)}=a_{s^{-1}(i)}$ and $b_{ij}^{(s)}=b_{s^{-1}(i),s^{-1}(j)}$,
so any bound expressed through symmetric averages of $(a_k,b_{k\ell})$ is independent of $s$.

\medskip
\noindent\emph{Step 1: $O_n$ concentrates.}
Since $W_i$ and $W_i'$ are conditionally independent given $\pi_0=s$,
\[
\E[O_n\mid\pi_0=s]=\frac1n\sum_k a_k^2=\rho_n^2+V_{1,n}\to\rho^2,
\]
using $\sum_k(a_k-\rho_n)^2=nV_{1,n}$.
For the variance, Lemma~\ref{lem:overlap-bound} reduces the problem to showing
\begin{equation}\label{eq:avg-cov-vanish}
\frac{1}{n^2}\sum_{k\neq\ell}|b_{k\ell}-a_ka_\ell|\to 0.
\end{equation}
Writing $a_k=\rho_n+d_k$, the triangle inequality gives
\[
|b_{k\ell}-a_ka_\ell|
\le |b_{k\ell}-\tau_n|+|\tau_n-\rho_n^2|+|a_ka_\ell-\rho_n^2|.
\]
For the first, Cauchy--Schwarz with $\#\{k\neq\ell\}=n(n-1)$ gives
\[
\frac1{n^2}\sum_{k\neq\ell}|b_{k\ell}-\tau_n|
\le\frac{\sqrt{n(n-1)}}{n^2}\sqrt{\sum_{k\neq\ell}(b_{k\ell}-\tau_n)^2}
=\frac{\sqrt{n(n-1)}}{n}\sqrt{V_{2,n}}\to 0.
\]
The second is $|\tau_n-\rho_n^2|=O(1/n)$.
The third uses $a_ka_\ell-\rho_n^2=\rho_n(d_k+d_\ell)+d_kd_\ell$, so
\[
\frac1{n^2}\sum_{k\neq\ell}|a_ka_\ell-\rho_n^2|
\le 2\rho_n\sqrt{V_{1,n}}+V_{1,n}\to 0
\]
by Cauchy--Schwarz applied to $\sum_k|d_k|\le\sqrt n\sqrt{nV_{1,n}}$.
Combining gives \eqref{eq:avg-cov-vanish}, hence $\Var(O_n\mid\pi_0=s)\to 0$,
and Chebyshev yields $O_n\toP\rho^2$.

\medskip
\noindent\emph{Step 2: $O_n^{(0)}$ concentrates.}
Note that $O_n^{(0)}$ is a function of a single draw $\tau\sim q_n^{(s)}$.
The conditional mean is
\[
\E[O_n^{(0)}\mid\pi_0=s]
=\frac1n\sum_{i\in A_n}a_i^{(s)}
=\frac1n\sum_{k\in s^{-1}(A_n)}a_k
=\rho_n^2+\frac1n\sum_{k\in s^{-1}(A_n)}d_k.
\]
By Cauchy--Schwarz applied to the index set $s^{-1}(A_n)$ of cardinality $n_1$,
\[
\Big|\frac1n\sum_{k\in s^{-1}(A_n)}d_k\Big|
\le\frac{\sqrt{n_1}}{n}\sqrt{nV_{1,n}}
=\sqrt{\rho_n}\sqrt{V_{1,n}}\to 0,
\]
with a bound independent of $s$.
For the variance, $O_n^{(0)}=\frac1n\sum_{i\in A_n}W_i$ involves a single draw, so
\[
\Var(O_n^{(0)}\mid\pi_0=s)
=\frac1{n^2}\sum_{i\in A_n}\Var(W_i\mid\pi_0=s)
+\frac1{n^2}\sum_{\substack{i,j\in A_n\\ i\neq j}}\Cov(W_i,W_j\mid\pi_0=s).
\]
The diagonal contribution is at most $n_1/(4n^2)\le 1/(4n)$.
Each covariance satisfies 
$|\Cov(W_i,W_j\mid\pi_0=s)|=|b_{ij}^{(s)}-a_i^{(s)}a_j^{(s)}|$,
and the off-diagonal sum is dominated by 
$\frac1{n^2}\sum_{i\neq j}|b_{ij}^{(s)}-a_i^{(s)}a_j^{(s)}|$,
which vanishes by \eqref{eq:avg-cov-vanish} with bounds independent of $s$.
\end{proof}

\begin{proof}[Proof of Theorem~\ref{thm:cond-normal}]
Conditional on $\pi_0$, draw two independent randomizations $\tau,\tau'\iid q_n^{(\pi_0)}$
and let $W:=W(\tau)$, $W':=W(\tau')$, $T_n:=T_n(X,W)$, $T_n':=T_n(X,W')$.
By construction, $(T_n,T_n')$ is conditionally i.i.d.\ given $(X,\pi_0)$.

\medskip
\noindent\emph{Step 1: signal--noise decomposition.}
Under Assumption~\ref{ass:pitman}, substituting $X_i=\mu+\frac{\delta}{\sqrt n}w_{0,i}+\varepsilon_i$
into~\eqref{eq:Tn} gives three contributions.
The location $\mu$ cancels since $\sum_i(W_i-\rho_n)=0$.
The signal contributes
\[
\frac{1}{\sigma_\varepsilon\sqrt{n\rho_n(1-\rho_n)}}\cdot\frac{\delta}{\sqrt n}
\sum_{i=1}^n(W_i-\rho_n)w_{0,i},
\]
and since $w_{0,i}=\ind\{i\in A_n\}$ and $\sum_{i\in A_n}W_i=nO_n^{(0)}$,
\[
\sum_{i=1}^n(W_i-\rho_n)w_{0,i}=\sum_{i\in A_n}W_i-\rho_n n_1=n(O_n^{(0)}-\rho_n^2).
\]
Dividing by the normalizer yields $T_n=N_n+L_n$ where
\[
N_n:=\frac{1}{\sigma_\varepsilon\sqrt{n\rho_n(1-\rho_n)}}\sum_{i=1}^n(W_i-\rho_n)\varepsilon_i,
\qquad
L_n:=\frac{\delta}{\sigma_\varepsilon\sqrt{\rho_n(1-\rho_n)}}(O_n^{(0)}-\rho_n^2).
\]
By Proposition~\ref{prop:overlap}, $O_n^{(0)}-\rho_n^2\toP 0$ with bounds independent of $\pi_0$,
so $L_n\toP 0$ and likewise $L_n'\toP 0$.

\medskip
\noindent\emph{Step 2: conditional bivariate CLT for $(N_n,N_n')$.}
Conditional on $(W,W',\pi_0)$, write
$(N_n,N_n')=\sum_{i=1}^n Y_{n,i}$ with
\[
Y_{n,i}:=\frac{\varepsilon_i}{\sigma_\varepsilon\sqrt{n\rho_n(1-\rho_n)}}(W_i-\rho_n,\;W_i'-\rho_n).
\]
These are independent across $i$.
Using $W_i\in\{0,1\}$ (so $W_i^2=W_i$) and $\sum_i W_i=n_1=n\rho_n$,
\[
\sum_i(W_i-\rho_n)^2=\sum_i W_i-2\rho_n\sum_i W_i+n\rho_n^2=n\rho_n(1-\rho_n),
\]
giving $\Var(N_n\mid W,W',\pi_0)=\Var(N_n'\mid W,W',\pi_0)=1$. Similarly,
\[
\Cov(N_n,N_n'\mid W,W',\pi_0)
=\frac{1}{n\rho_n(1-\rho_n)}\sum_{i=1}^n(W_i-\rho_n)(W_i'-\rho_n)
=\frac{O_n-\rho_n^2}{\rho_n(1-\rho_n)}=:\gamma_n,
\]
using $\sum_i W_iW_i'=nO_n$.
The conditional covariance matrix is 
$\Sigma_n=\bigl(\begin{smallmatrix}1&\gamma_n\\\gamma_n&1\end{smallmatrix}\bigr)$,
with $\gamma_n\toP 0$ by Proposition~\ref{prop:overlap}.

\emph{Lindeberg condition.}
Since $|W_i-\rho_n|\le 1$, we have 
$\|Y_{n,i}\|^2\le 2\varepsilon_i^2/[\sigma_\varepsilon^2 n\rho_n(1-\rho_n)]$,
so the event $\|Y_{n,i}\|>\eta$ forces 
$|\varepsilon_i|>c_\eta\sqrt n$ with 
$c_\eta:=\eta\sigma_\varepsilon\sqrt{\rho_n(1-\rho_n)/2}$ (bounded away from $0$).
Therefore
\[
\sum_{i=1}^n\E\big[\|Y_{n,i}\|^2\ind\{\|Y_{n,i}\|>\eta\}\mid W,W',\pi_0\big]
\le\frac{2\,\E[\varepsilon_1^2\ind\{|\varepsilon_1|>c_\eta\sqrt n\}]}{\sigma_\varepsilon^2\rho_n(1-\rho_n)}
\to 0,
\]
using $\E|\varepsilon_1|^{2+\kappa}<\infty$ 
(so $\E[\varepsilon_1^2\ind\{|\varepsilon_1|>c_\eta\sqrt n\}]
\le\E|\varepsilon_1|^{2+\kappa}/(c_\eta\sqrt n)^\kappa\to 0$).

The bivariate Lindeberg--Feller CLT therefore gives, for every bounded Lipschitz function $g:\mathbb R^2\to\mathbb R$,
\[
\E\{g(N_n,N_n')\mid W,W',\pi_0\}
-
\int g\,d\N(0,\Sigma_n)
\toP0.
\]
Since $\gamma_n\toP0$, the Gaussian laws $\N(0,\Sigma_n)$ converge weakly to the product law of two independent standard normals. Hence
\[
\int g\,d\N(0,\Sigma_n)\toP\E g(Z,Z'),
\]
where $Z,Z'$ are i.i.d.\ $\N(0,1)$.
Taking conditional expectations with respect to $\pi_0$ gives
\[
\E\{g(N_n,N_n')\mid \pi_0\}\toP\E g(Z,Z').
\]
Thus the conditional law of $(N_n,N_n')$ given $\pi_0$ converges weakly in probability to the product standard normal law.

\medskip
\noindent\emph{Step 3: conclusion.}
Combining Steps~1--2 with Slutsky's theorem gives 
$(T_n,T_n')\mid\pi_0\dto(Z,Z')$ with $Z,Z'$ i.i.d.\ $\N(0,1)$.
Applying Lemma~\ref{lem:two-rand} with $F=\Phi$ yields
\[
\sup_t|F_n(t\mid X,\pi_0)-\Phi(t)|\toP 0.
\]
Since $\Phi$ is continuous and strictly increasing, 
uniform convergence of cdfs implies pointwise convergence of the inverses:
$c_n(\alpha\mid X,\pi_0)=F_n^{-1}(1-\alpha\mid X,\pi_0)
\toP\Phi^{-1}(1-\alpha)=z_{1-\alpha}$.
\end{proof}

\section{Proof of Proposition~\ref{prop:obs}}\label{app:proof-obs}

\begin{proof}
Substitute $X_i=\mu+\frac{\delta}{\sqrt n}w_{0,i}+\varepsilon_i$ into $T_n^{\mathrm{obs}}=T_n(X,w_0)$.
Since $\sum_i(w_{0,i}-\rho_n)=0$, the constant $\mu$ cancels.

For the signal term, $w_{0,i}^2=w_{0,i}$ (since $w_{0,i}\in\{0,1\}$) and 
$\sum_i w_{0,i}=n_1=n\rho_n$ give
\[
\sum_{i=1}^n(w_{0,i}-\rho_n)w_{0,i}
=\sum_i w_{0,i}^2-\rho_n\sum_i w_{0,i}
=n_1-\rho_n n_1=n\rho_n(1-\rho_n),
\]
hence the signal term equals
\[
\frac{1}{\sigma_\varepsilon\sqrt{n\rho_n(1-\rho_n)}}\cdot\frac{\delta}{\sqrt n}\cdot n\rho_n(1-\rho_n)
=\frac{\delta\sqrt{\rho_n(1-\rho_n)}}{\sigma_\varepsilon}\to\Delta.
\]

For the noise term, set 
$c_{n,i}:=(w_{0,i}-\rho_n)/\sqrt{n\rho_n(1-\rho_n)}$,
so that $\sum_i c_{n,i}^2=1$ (by the same computation as above, 
applied to $\sum_i(w_{0,i}-\rho_n)^2$) 
and $\max_i|c_{n,i}|\le 1/\sqrt{n\rho_n(1-\rho_n)}=O(n^{-1/2})\to 0$.
Since $\varepsilon_i$ are i.i.d.\ with $\Var(\varepsilon_i)=\sigma_\varepsilon^2$ and 
$\E|\varepsilon_i|^{2+\kappa}<\infty$,
the Lindeberg--Feller CLT gives 
$\sum_i c_{n,i}\varepsilon_i/\sigma_\varepsilon\dto\N(0,1)$.
Combining the deterministic signal limit with Slutsky's theorem yields 
$T_n^{\mathrm{obs}}\dto\N(\Delta,1)$.
\end{proof}

\section{Proof of Theorem~\ref{thm:necessary}}\label{app:proof-necessary}

The key ingredient is the following non-self-averaging result.

\begin{proposition}[Non-self-averaging under Gaussian noise]\label{prop:nonself}
Under Assumptions~\ref{ass:rhon}, \ref{ass:pitman}, and~\ref{ass:gauss-noise},
recall from Appendix~\ref{app:proof-cond-normal} the decomposition 
$T_n=N_n+\lambda_n D_n$ (and likewise $T_n'=N_n'+\lambda_n D_n'$), where
\[
\gamma_n:=\frac{O_n-\rho_n^2}{\rho_n(1-\rho_n)},\quad
D_n:=O_n^{(0)}-\rho_n^2,\quad
D_n':=O_n^{(0)\prime}-\rho_n^2,\quad
\lambda_n:=\frac{\delta}{\sigma_\varepsilon\sqrt{\rho_n(1-\rho_n)}},
\]
and $O_n^{(0)\prime}:=\frac1n\sum_{i\in A_n}W_i(\tau')$ is the reference overlap 
for the second independent randomization $\tau'$.
Then
\[
\Var\bigl(\E[T_n^2\mid X,\pi_0]\bigr)
=2\,\E[\gamma_n^2]+4\lambda_n^2\,\E[\gamma_n D_nD_n']
+\lambda_n^4\,\Var\bigl(\E[D_n^2\mid\pi_0]\bigr),
\]
with $\E[\gamma_n D_nD_n']\ge 0$.
In particular, $\Var(\E[T_n^2\mid X,\pi_0])\ge 2\,\E[\gamma_n^2]$.
\end{proposition}

\begin{proof}
Since $(T_n,T_n')$ is i.i.d.\ conditional on $(X,\pi_0)$,
\[
\Var\bigl(\E[T_n^2\mid X,\pi_0]\bigr)
=\E[T_n^2(T_n')^2]-\E[T_n^2]^2.
\]

\medskip
\noindent\emph{Step 1: fourth-moment identity.}
Under Gaussian noise, conditional on $(W,W',\pi_0)$, the pair 
$(T_n,T_n')=(N_n+\lambda_n D_n,\,N_n'+\lambda_n D_n')$ is bivariate normal 
with means $(\lambda_n D_n,\lambda_n D_n')$, unit marginal variances, 
and correlation $\gamma_n$ (from Appendix~\ref{app:proof-cond-normal}, Step~2).
Write $U:=T_n-\lambda_n D_n$ and $U':=T_n'-\lambda_n D_n'$, so that $(U,U')$
is centered bivariate normal with covariance 
$\bigl(\begin{smallmatrix}1&\gamma_n\\\gamma_n&1\end{smallmatrix}\bigr)$.
Isserlis' theorem gives $\E[U^2(U')^2\mid W,W',\pi_0]=1+2\gamma_n^2$.
Expanding $(U+\lambda_n D_n)^2(U'+\lambda_n D_n')^2$ and discarding all terms 
containing odd powers of $U$ or $U'$ (which vanish in expectation) yields
\[
\E[T_n^2(T_n')^2\mid W,W',\pi_0]
=1+2\gamma_n^2+\lambda_n^2\bigl(D_n^2+(D_n')^2\bigr)+4\lambda_n^2\gamma_n D_nD_n'+\lambda_n^4 D_n^2(D_n')^2.
\]

\medskip
\noindent\emph{Step 2: assembling the variance.}
Taking unconditional expectations and using that $D_n,D_n'$ are identically distributed,
\[
\E[T_n^2(T_n')^2]
=1+2\E[\gamma_n^2]+2\lambda_n^2\E[D_n^2]+4\lambda_n^2\E[\gamma_n D_nD_n']+\lambda_n^4\E[D_n^2(D_n')^2].
\]
On the other hand, $\E[T_n^2]=1+\lambda_n^2\E[D_n^2]$, so 
$\E[T_n^2]^2=1+2\lambda_n^2\E[D_n^2]+\lambda_n^4\E[D_n^2]^2$.
Subtracting,
\[
\Var\bigl(\E[T_n^2\mid X,\pi_0]\bigr)
=2\E[\gamma_n^2]+4\lambda_n^2\E[\gamma_n D_nD_n']+\lambda_n^4\bigl(\E[D_n^2(D_n')^2]-\E[D_n^2]^2\bigr).
\]
Since $D_n,D_n'$ are conditionally i.i.d.\ given $\pi_0$,
$\E[D_n^2(D_n')^2]=\E[\E[D_n^2\mid\pi_0]^2]$, 
hence the last term equals $\lambda_n^4\Var(\E[D_n^2\mid\pi_0])\ge 0$.

\medskip
\noindent\emph{Step 3: sign of $\E[\gamma_n D_nD_n']$.}
Write $u_i:=W_i-\rho_n$ and $u_i':=W_i'-\rho_n$.
Using $W_i\in\{0,1\}$ and $\sum_i W_i=n_1$, direct computation gives
\[
\gamma_n=\frac{1}{n\rho_n(1-\rho_n)}\sum_{i=1}^n u_iu_i',
\qquad
D_n=\frac{1}{n}\sum_{j\in A_n}u_j,
\qquad
D_n'=\frac{1}{n}\sum_{j\in A_n}u_j'.
\]
Since $W,W'$ are conditionally independent given $\pi_0$ 
and share the same conditional law $q_n^{(\pi_0)}$,
$\E[u_iu_i'u_ju_k'\mid\pi_0]=\E[u_iu_j\mid\pi_0]\cdot\E[u_iu_k\mid\pi_0]$.
Multiplying out $\gamma_n D_n D_n'$ and taking conditional expectations,
\[
\E[\gamma_n D_nD_n'\mid\pi_0]
=\frac{1}{n^3\rho_n(1-\rho_n)}\sum_{i=1}^n
\Big(\sum_{j\in A_n}\E[u_iu_j\mid\pi_0]\Big)^2\ge 0,
\]
and taking outer expectation over $\pi_0$ preserves the sign.

Combining Steps~1--3, the three terms in the stated identity are nonnegative
(the last two by Steps~2 and~3), yielding 
$\Var(\E[T_n^2\mid X,\pi_0])\ge 2\E[\gamma_n^2]$.
\end{proof}

\begin{proof}[Proof of Theorem~\ref{thm:necessary}]
\emph{Step 1: reduction to a subsequence with $\liminf\E[\gamma_n^2]>0$.}
Suppose $V_{1,n}+V_{2,n}\not\to 0$, so $\limsup_n(V_{1,n}+V_{2,n})>0$.
Then there exists a subsequence, still denoted by $n_k$, and a constant $c>0$ such that
\[
V_{1,n_k}+V_{2,n_k}\ge c
\qquad\text{for all }k.
\]
Passing to a further subsequence if necessary, one of the following two cases holds:
either $\liminf_k V_{1,n_k}>0$, or $V_{1,n_k}\to0$.
In the first case, Lemma~\ref{lem:overlap-V} gives
\[
\E[(O_{n_k}-\rho_{n_k}^2)^2]\ge V_{1,n_k}^2\not\to0.
\]
In the second case, the lower bound $V_{1,n_k}+V_{2,n_k}\ge c$ implies
$\liminf_k V_{2,n_k}\ge c$, and Lemma~\ref{lem:overlap-V} gives
\[
\E[(O_{n_k}-\rho_{n_k}^2)^2]
=
V_{2,n_k}+o(1)
\not\to0.
\]
In both cases, $\liminf_k\E[(O_{n_k}-\rho_{n_k}^2)^2]>0$.

Since $\rho_n\to\rho\in(0,1)$, the denominator $\rho_n^2(1-\rho_n)^2$ 
converges to $\rho^2(1-\rho)^2>0$, so
\[
\liminf_k\E[\gamma_{n_k}^2]
=\liminf_k\frac{\E[(O_{n_k}-\rho_{n_k}^2)^2]}{\rho_{n_k}^2(1-\rho_{n_k})^2}>0.
\]

\medskip
\noindent\emph{Step 2: non-self-averaging.}
By Proposition~\ref{prop:nonself}, 
$\Var(\E[T_n^2\mid X,\pi_0])\ge 2\,\E[\gamma_n^2]$, 
so $\liminf_k\Var(\E[T_{n_k}^2\mid X,\pi_0])>0$ along the same subsequence.

\medskip
\noindent\emph{Step 3: no deterministic cdf limit.}
We show, by contradiction, that $F_n(\cdot\mid X,\pi_0)$ cannot converge in probability
to any deterministic continuous cdf.
Suppose $F_n(\cdot\mid X,\pi_0)\toP F$ for some deterministic continuous cdf $F$
with $m_2:=\int t^2\,dF(t)<\infty$.
Under Assumption~\ref{ass:gauss-noise}, conditional on $(W,W',\pi_0)$ the statistic $T_n$
is Gaussian with bounded variance (Appendix~\ref{app:proof-cond-normal}, Step~2),
so $\{T_n^2\}$ is uniformly integrable.
We justify the moment implication for these random conditional laws.
Let $H_n(\cdot):=F_n(\cdot\mid X,\pi_0)$.
For each fixed $K<\infty$, the function $t\mapsto t^2\wedge K$ is bounded and continuous, so
\[
\int (t^2\wedge K)\,dH_n(t)
\toP
\int (t^2\wedge K)\,dF(t).
\]
It remains to control the tails. Since, conditional on $(W,\pi_0)$, $T_n$ is Gaussian with unit variance and uniformly bounded mean, the sequence $\{T_n^2\}$ is uniformly integrable. Hence
\[
\E\left[
\int t^2\ind\{|t|>K\}\,dH_n(t)
\right]
=
\E\left[T_n^2\ind\{|T_n|>K\}\right]
\to 0
\]
as $K\to\infty$, uniformly in $n$.
By Markov's inequality, the same tail term is $o_P(1)$ after choosing $K$ large.
Letting $K\to\infty$ gives
\[
\E[T_n^2\mid X,\pi_0]=
\int t^2\,dH_n(t)
\toP
\int t^2\,dF(t)=m_2.
\]
Therefore $\Var(\E[T_n^2\mid X,\pi_0])\to0$.
This contradicts Step~2.

\medskip
\noindent\emph{Step 4: a non-stabilizing critical value.}
Suppose, for contradiction, that $c_n(\alpha\mid X,\pi_0)\toP z_{1-\alpha}$ 
for every $\alpha\in(0,1)$.
Indeed, fix $\varepsilon>0$ and choose grid points
$0<u_1<\cdots<u_J<1$ such that $u_{j+1}-u_j<\varepsilon$,
$u_1<\varepsilon$, and $1-u_J<\varepsilon$.
By the assumed quantile convergence,
\[
F_n^{-1}(u_j\mid X,\pi_0)\toP \Phi^{-1}(u_j),
\qquad j=1,\ldots,J.
\]
On the event that these finitely many grid quantiles are close to their limits,
monotonicity of $F_n$ and uniform continuity of $\Phi$ imply
\[
\sup_{t\in\mathbb R}|F_n(t\mid X,\pi_0)-\Phi(t)|
\le C\varepsilon+o(1).
\]
Since $\varepsilon$ is arbitrary, the uniform convergence follows.
This contradicts Step~3, so there exists $\alpha\in(0,1)$ for which
$c_n(\alpha\mid X,\pi_0)\not\toP z_{1-\alpha}$.
\end{proof}


\section{Proof of Theorem~\ref{thm:AL}}\label{app:proof-AL}

\begin{proof}
Let
\[
G_n(t\mid X,\pi_0):=\mP(S_n(X,W(\tau))\le t\mid X,\pi_0),
\qquad
G_n^L(t\mid X,\pi_0):=\mP(L_n(W(\tau))\le t\mid X,\pi_0)
\]
denote the conditional permutation cdfs of $S_n$ and its ideal linear approximation.
Let $d_n(\alpha\mid X,\pi_0):=G_n^{-1}(1-\alpha\mid X,\pi_0)$.

\medskip
\noindent\emph{Step 1: Theorems~\ref{thm:cond-normal}--\ref{thm:power} apply to $L_n$.}
By Assumption~\ref{ass:AL}, the variables
$\mu_L+\beta w_{0,i}/\sqrt n+\xi_i$ satisfy the same Pitman local structure as in
Assumption~\ref{ass:pitman}, with
$(\mu,\delta,\varepsilon_i,\sigma_\varepsilon)$ replaced by
$(\mu_L,\beta,\xi_i,\sigma_L)$.
Assumptions~\ref{ass:rhon} and~\ref{ass:V} are unchanged.
Therefore Theorem~\ref{thm:cond-normal} yields
\[
\sup_t |G_n^L(t\mid X,\pi_0)-\Phi(t)|\toP0,
\]
and Proposition~\ref{prop:obs} gives
$L_n(w_0)\dto\N(\Delta_L,1)$.

\medskip
\noindent\emph{Step 2: transfer of the conditional permutation cdf.}
For fixed $\eta>0$, define the conditional approximation error
\[
A_n(\eta\mid X,\pi_0)
:=\mP_{\tau\sim q_n^{(\pi_0)}}\Bigl(
\bigl|S_n\{X,W(\tau)\}-L_n\{W(\tau)\}\bigr|>\eta
\ \Bigm|\ X,\pi_0\Bigr).
\]
By Assumption~\ref{ass:AL}, $A_n(\eta\mid X,\pi_0)\toP0$ for every $\eta>0$.
On the event $|S_n\{X,W(\tau)\}-L_n\{W(\tau)\}|\le\eta$,
\[
\{L_n(W(\tau))\le t-\eta\}\subseteq\{S_n(X,W(\tau))\le t\}
\subseteq\{L_n(W(\tau))\le t+\eta\}.
\]
Taking conditional probabilities gives, for all $t$,
\[
G_n^L(t-\eta\mid X,\pi_0)-A_n(\eta\mid X,\pi_0)
\le
G_n(t\mid X,\pi_0)
\le
G_n^L(t+\eta\mid X,\pi_0)+A_n(\eta\mid X,\pi_0).
\]
Hence
\[
\sup_t |G_n(t\mid X,\pi_0)-\Phi(t)|
\le
\sup_u|G_n^L(u\mid X,\pi_0)-\Phi(u)|
+\sup_t\{\Phi(t+\eta)-\Phi(t-\eta)\}
+A_n(\eta\mid X,\pi_0).
\]
The first and third terms converge to zero in probability for each fixed $\eta>0$, while
$\sup_t\{\Phi(t+\eta)-\Phi(t-\eta)\}\to0$ as $\eta\downarrow0$.
Therefore $\sup_t|G_n(t\mid X,\pi_0)-\Phi(t)|\toP0$.
The usual quantile-continuity argument then implies
$d_n(\alpha\mid X,\pi_0)\toP z_{1-\alpha}$.

\medskip
\noindent\emph{Step 3: limiting observed statistic and power.}
Assumption~\ref{ass:AL} also gives $S_n(X,w_0)-L_n(w_0)\toP0$.
Since $L_n(w_0)\dto\N(\Delta_L,1)$ by Step~1, Slutsky's theorem yields
$S_n^{\mathrm{obs}}\dto\N(\Delta_L,1)$.
Combining this with the critical-value convergence from Step~2, as in the proof of
Theorem~\ref{thm:power}, gives
\[
\mP\bigl(S_n^{\mathrm{obs}}>d_n(\alpha\mid X,\pi_0)\bigr)
\to 1-\Phi(z_{1-\alpha}-\Delta_L). \qedhere
\]
\end{proof}

\section{Proof of Proposition~\ref{prop:ex-cyclic}}\label{app:proof-cyclic}

\begin{proof}
Let $c:[n]\to[n]$ be the cyclic shift $c(i)=(i\bmod n)+1$, 
so that $\pi=c^J$ with $J\sim\mathrm{Uniform}\{0,\dots,n-1\}$ gives
$\pi(k)=((k-1+J)\bmod n)+1$.
The event $\{\pi(k)\in A_n\}=\{\pi(k)\le n_1\}$ is equivalent to $U\le n_1-1$,
where $U:=(k-1+J)\bmod n$.

\medskip
\noindent\emph{Step 1: $V_{1,n}=0$.}
For each fixed $k$, $J\mapsto(k-1+J)\bmod n$ is a bijection on $\{0,\dots,n-1\}$,
so $U\sim\mathrm{Uniform}\{0,\dots,n-1\}$ regardless of $k$.
Hence $a_k=\mP(U\le n_1-1)=n_1/n=\rho_n$ for all $k$, giving $V_{1,n}=0$.

\medskip
\noindent\emph{Step 2: $b_{k\ell}$ depends only on the cyclic distance.}
For $k\neq\ell$, set $d:=(\ell-k)\bmod n\in\{1,\dots,n-1\}$.
Since $\pi(\ell)=((\ell-1+J)\bmod n)+1=((U+d)\bmod n)+1$, the joint event is
\[
\{\pi(k)\in A_n,\,\pi(\ell)\in A_n\}
=\{U\le n_1-1\}\cap\{(U+d)\bmod n\le n_1-1\},
\]
whose probability depends on $(k,\ell)$ only through $d$.
Define
$\beta(d):=\frac1n|\{u\in\{0,\dots,n_1-1\}:(u+d)\bmod n\le n_1-1\}|$,
so that $b_{k\ell}=\beta(d)$.

\medskip
\noindent\emph{Step 3: closed form for $\beta(d)$.}
Count $u\in\{0,\dots,n_1-1\}$ with $(u+d)\bmod n\le n_1-1$ by two cases.
\emph{Without wraparound} ($u+d<n$): the condition reads $u+d\le n_1-1$, i.e., 
$u\le n_1-d-1$. Combined with $u\ge 0$, this gives $u\in\{0,\dots,n_1-d-1\}$, 
contributing $(n_1-d)^+$ values.
\emph{With wraparound} ($u+d\ge n$): the condition reads $u+d-n\le n_1-1$, i.e., $u\le n_1-1-d+n$.
Combined with $u\ge n-d$ and $u\le n_1-1$, this gives $u\in\{n-d,\dots,n_1-1\}$, 
which is nonempty only when $n-d\le n_1-1$ (i.e., $d\ge n_0+1$), 
contributing $(d-n_0)^+$ values.
Therefore
\[
\beta(d)=\frac{(n_1-d)^++(d-n_0)^+}{n}.
\]

\medskip
\noindent\emph{Step 4: symmetry $V_{2,n}(n_1,n_0)=V_{2,n}(n_0,n_1)$.}
We show that the centered quantity $b_{k\ell}-\tau_n$ is invariant under $n_1\leftrightarrow n_0$.
By inclusion--exclusion,
\[
b_{k\ell}^{(n_0,n_1)}
=\mP(\pi(k)\notin A_n,\,\pi(\ell)\notin A_n)
=1-2\rho_n+b_{k\ell}^{(n_1,n_0)}.
\]
On the other hand, a direct computation gives
$\tau_n^{(n_0,n_1)}-\tau_n^{(n_1,n_0)}
=\frac{n_0(n_0-1)-n_1(n_1-1)}{n(n-1)}=\frac{n_0-n_1}{n}$.
Subtracting,
\[
\bigl(b_{k\ell}^{(n_0,n_1)}-\tau_n^{(n_0,n_1)}\bigr)
-\bigl(b_{k\ell}^{(n_1,n_0)}-\tau_n^{(n_1,n_0)}\bigr)
=(1-2\rho_n)-\frac{n_0-n_1}{n}=0,
\]
so the centered quantities agree and $V_{2,n}(n_0,n_1)=V_{2,n}(n_1,n_0)$.
It therefore suffices to compute $V_{2,n}$ when $n_*:=\min(n_1,n_0)=n_1\le n_0=:n^*$.

\medskip
\noindent\emph{Step 5: closed form for $V_{2,n}$.}
Each cyclic distance $d\in\{1,\dots,n-1\}$ corresponds to exactly $n$ ordered pairs 
$(k,\ell)$ with $(\ell-k)\bmod n=d$, so
\[
V_{2,n}=\frac1n\sum_{d=1}^{n-1}\bigl(\beta(d)-\tau_n\bigr)^2.
\]
When $n_1\le n_0$, Step~3 gives the three-region formula
\[
\beta(d)=\begin{cases}
(n_1-d)/n & 1\le d\le n_1-1,\\
0 & n_1\le d\le n_0,\\
(d-n_0)/n & n_0+1\le d\le n-1.
\end{cases}
\]
The first and third regions are mirror images: substituting $j=n_1-d$ in the first 
and $j=d-n_0$ in the third both give $j\in\{1,\dots,n_1-1\}$ with $\beta=j/n$,
so their contributions sum to $\frac2n\sum_{j=1}^{n_1-1}(j/n-\tau_n)^2$.
The middle region contributes $n_0-n_1+1$ terms with $\beta=0$, giving 
$\frac{n_0-n_1+1}{n}\tau_n^2$.  Therefore
\[
V_{2,n}=\frac{2}{n}\sum_{j=1}^{n_1-1}\Big(\frac{j}{n}-\tau_n\Big)^2
+\frac{n_0-n_1+1}{n}\tau_n^2.
\]
Expanding the square,
\[
\sum_{j=1}^{n_1-1}\Big(\frac{j}{n}-\tau_n\Big)^2
=\frac{1}{n^2}\sum_{j=1}^{n_1-1}j^2
-\frac{2\tau_n}{n}\sum_{j=1}^{n_1-1}j
+(n_1-1)\tau_n^2,
\]
and applying $\sum_{j=1}^{m}j=m(m+1)/2$, $\sum_{j=1}^{m}j^2=m(m+1)(2m+1)/6$ with $m=n_1-1$,
then substituting $\tau_n=n_1(n_1-1)/[n(n-1)]$ and combining over the common denominator 
$3n^3(n-1)$ yields
\[
V_{2,n}=\frac{n_1(n_1-1)\bigl(2n_1n_0-n_1^2-n_0+1\bigr)}{3n^3(n-1)}.
\]
By Step~4, the formula extends to $n_1>n_0$ by replacing $(n_1,n_0)$ with $(n_*,n^*)$.

\medskip
\noindent\emph{Step 6: asymptotic limit.}
With $n_*/n\to\rho_*$ and $n^*/n\to 1-\rho_*$, the leading-order behavior is
\[
V_{2,n}
\sim\frac{n_*^2\cdot(2n_*n^*-n_*^2)}{3n^4}
=\frac{(n_*/n)^2\bigl[2(n_*/n)(n^*/n)-(n_*/n)^2\bigr]}{3}
\to\frac{\rho_*^2\bigl[2\rho_*(1-\rho_*)-\rho_*^2\bigr]}{3}
=\rho_*^{\,3}\Big(\frac23-\rho_*\Big),
\]
which is strictly positive since $\rho_*\le 1/2<2/3$.
\end{proof}

\section{Proof of Proposition~\ref{prop:ex-block}}\label{app:proof-block}

\begin{proof}
Recall the block setup: $p:=n_1/2$ treated indices in each of $B_1,B_2\subseteq A_n$, 
$q:=n_0/2$ control indices in each of $B_3,B_4\subseteq A_n^c$, 
and $N:=n/2$ indices in each cross-arm pair 
$\mathcal P_1=B_1\cup B_3,\mathcal P_2=B_2\cup B_4$.
The randomization $\pi$ applies independent uniform permutations within $\mathcal P_1$ and $\mathcal P_2$.

\medskip
\noindent\emph{Step 1: $V_{1,n}=0$.}
For any $k\in\mathcal P_j$, $\pi(k)$ is uniform on $\mathcal P_j$, 
which contains $p$ treated indices.
Hence $a_k=|A_n\cap\mathcal P_j|/|\mathcal P_j|=p/N=\rho_n$ for every $k$, giving $V_{1,n}=0$.

\medskip
\noindent\emph{Step 2: two-point block probabilities.}
If $k,\ell$ lie in the same pair $\mathcal P_j$, the ordered pair $(\pi(k),\pi(\ell))$
is a uniform sample of size two without replacement from $\mathcal P_j$, giving
\[
b^{\mathrm{same}}
=\frac{|A_n\cap\mathcal P_j|\cdot(|A_n\cap\mathcal P_j|-1)}{N(N-1)}
=\frac{p(p-1)}{N(N-1)}.
\]
If $k\in\mathcal P_1$ and $\ell\in\mathcal P_2$, the permutations within the two pairs
are independent by construction, so $\pi(k)$ and $\pi(\ell)$ are independent, 
each uniform on its pair. Hence
\[
b^{\mathrm{diff}}=\mP(\pi(k)\in A_n)\cdot\mP(\pi(\ell)\in A_n)=\frac{p^2}{N^2}.
\]

\medskip
\noindent\emph{Step 3: deviations from $\tau_n$.}
Using $n_1=2p$ and $n=2N$, 
$\tau_n=n_1(n_1-1)/[n(n-1)]=p(2p-1)/[N(2N-1)]$.
Direct algebra over the common denominator $N(N-1)(2N-1)$ gives
\[
b^{\mathrm{same}}-\tau_n
=\frac{p(p-1)(2N-1)-p(2p-1)(N-1)}{N(N-1)(2N-1)}
=\frac{-pq}{N(N-1)(2N-1)},
\]
where the last equality uses $(p-1)(2N-1)-(2p-1)(N-1)=p-N=-q$.
Similarly,
\[
b^{\mathrm{diff}}-\tau_n
=\frac{p^2(2N-1)-p(2p-1)N}{N^2(2N-1)}
=\frac{pq}{N^2(2N-1)}.
\]
Both are $O(1/n)$.

\medskip
\noindent\emph{Step 4: counting pairs and assembling $V_{2,n}$.}
There are $C^{\mathrm{same}}=2N(N-1)$ ordered same-pair pairs (two pairs, each with $N(N-1)$)
and $C^{\mathrm{diff}}=2N^2$ ordered cross-pair pairs.
Since $b_{k\ell}-\tau_n$ takes only these two values,
\[
V_{2,n}
=\frac{C^{\mathrm{same}}}{n^2}(b^{\mathrm{same}}-\tau_n)^2
+\frac{C^{\mathrm{diff}}}{n^2}(b^{\mathrm{diff}}-\tau_n)^2.
\]
Substituting from Steps~2--3:
\begin{align*}
V_{2,n}
&=\frac{2N(N-1)}{n^2}\cdot\frac{p^2q^2}{N^2(N-1)^2(2N-1)^2}
+\frac{2N^2}{n^2}\cdot\frac{p^2q^2}{N^4(2N-1)^2}\\
&=\frac{2p^2q^2}{n^2 N(N-1)(2N-1)^2}+\frac{2p^2q^2}{n^2 N^2(2N-1)^2}
=\frac{2p^2q^2}{n^2 N^2(N-1)(2N-1)^2}\bigl[N+(N-1)\bigr]\\
&=\frac{2p^2q^2}{n^2 N^2(N-1)(2N-1)}.
\end{align*}
Using $n=2N$ so $n^2=4N^2$, this simplifies to
\[
V_{2,n}=\frac{p^2q^2}{2N^4(N-1)(2N-1)}.
\]
Substituting $p=n_1/2$, $q=n_0/2$, $N=n/2$ (so $N-1=(n-2)/2$ and $2N-1=n-1$) gives
\[
V_{2,n}=\frac{n_1^2 n_0^2}{n^4(n-2)(n-1)}=O(n^{-2})\to 0. \qedhere
\]
\end{proof}

\section{Proofs for Section~\ref{sec:nonuniform}}\label{app:proof-nonuniform}

\subsection{Proof of Proposition~\ref{prop:var-full}}\label{app:proof-var-full}

\begin{proof}
\emph{Step 1: Permutation variance formula.}
With $n_1=n_0=n/2$, we have 
$\bar X_{\mathrm{trt}}=\frac{2}{n}\sum_i W_iX_i$ and 
$\bar X_{\mathrm{ctrl}}=\frac{2}{n}\sum_i(1-W_i)X_i$, so
\[
T=\frac{2}{n}\sum_i(2W_i-1)X_i
=\frac{4}{n}\sum_{i=1}^n\Big(W_i-\tfrac12\Big)X_i.
\]
Under Strategy~1, $W$ is a simple random sample indicator of size $n/2$, giving
$\Var(W_i\mid X)=1/4$ and $\Cov(W_i,W_j\mid X)=-1/[4(n-1)]$ for $i\neq j$.
A standard computation then yields
\[
\Var\Big(\sum_i(W_i-\tfrac12)X_i\,\Big|\,X\Big)
=\frac{n}{4(n-1)}\sum_i(X_i-\bar X)^2,
\]
and multiplying by $(4/n)^2$ gives
\[
\Var_{\mathrm{full}}(T\mid X)=\frac{4}{n(n-1)}\sum_{i=1}^n(X_i-\bar X)^2.
\]

\medskip
\noindent\emph{Step 2: Limit of $\frac1n\sum_i(X_i-\bar X)^2$ under $H_1$.}
Write $X_i=m_i+s_i+\varepsilon_i$ where $m_i:=\mu_{\mathrm{s}(i)}$ (stratum mean)
and $s_i:=(\delta/\sqrt n)w_{0,i}$ (local signal).
Since $\bar m=0$ (balanced strata) and $\bar s=\delta/(2\sqrt n)$,
$X_i-\bar X=m_i+(s_i-\bar s)+(\varepsilon_i-\bar\varepsilon)$.
Expanding $\frac1n\sum_i(X_i-\bar X)^2$:
\begin{align*}
\frac1n\sum_i(X_i-\bar X)^2
&=\underbrace{\frac1n\sum_i m_i^2}_{=\,\mu^2}
+\underbrace{\frac1n\sum_i(s_i-\bar s)^2}_{=\,\delta^2/(4n)}
+\underbrace{\frac1n\sum_i(\varepsilon_i-\bar\varepsilon)^2}_{\toP\,1}\\
&\quad+\underbrace{\frac2n\sum_i m_i(s_i-\bar s)}_{=\,O(1/\sqrt n)}
+\underbrace{\frac2n\sum_i m_i(\varepsilon_i-\bar\varepsilon)}_{\toP\,0}
+\underbrace{\frac2n\sum_i(s_i-\bar s)(\varepsilon_i-\bar\varepsilon)}_{=\,O_P(1/n)}.
\end{align*}
\emph{Dominant terms.}
$\frac1n\sum_i m_i^2=\frac{n/2}{n}\mu^2+\frac{n/2}{n}\mu^2=\mu^2$;
$\sum_i(s_i-\bar s)^2=(\delta^2/n)\sum_i(w_{0,i}-1/2)^2=(\delta^2/n)\cdot(n/4)=\delta^2/4$, 
hence $\frac1n\sum_i(s_i-\bar s)^2=\delta^2/(4n)\to 0$;
the noise variance converges to $\sigma_\varepsilon^2=1$ by the law of large numbers.

\emph{Cross terms.}
For the stratum--signal cross, Cauchy--Schwarz gives
$\big|\frac2n\sum_i m_i(s_i-\bar s)\big|\le\frac2n\sqrt{n\mu^2}\cdot\sqrt{\delta^2/4}=\mu\delta/\sqrt n\to 0$.
The stratum--noise cross $\frac2n\sum_i m_i\varepsilon_i$ is a mean-zero sum with variance
$(4\mu^2/n^2)\cdot n=4\mu^2/n\to 0$ (using $\sum_i m_i^2=n\mu^2$).
The signal--noise cross equals $\frac{2\delta}{n\sqrt n}\sum_i(w_{0,i}-1/2)\varepsilon_i$, 
a mean-zero sum with variance $(4\delta^2/n^3)\cdot(n/4)=\delta^2/n^2$, hence $O_P(1/n)$.

Combining, $\frac1n\sum_i(X_i-\bar X)^2\toP 1+\mu^2$, hence
$n\cdot\Var_{\mathrm{full}}(T\mid X)=\frac{4n}{n-1}\cdot\frac1n\sum_i(X_i-\bar X)^2\toP 4(1+\mu^2)$.
\end{proof}

\subsection{Proof of Proposition~\ref{prop:var-strat}}\label{app:proof-var-strat}

\begin{proof}
\emph{Step 1: decomposition $T=\frac12(T_1+T_2)$.}
With balanced design $n_{1,s}=n_{0,s}=n/4$ per stratum, the treated mean decomposes as
\[
\bar X_{\mathrm{trt}}
=\frac{1}{n/2}\sum_{s=1}^2\sum_{i\in\mathcal S_s:W_i=1}\!\!X_i
=\frac{1}{2}(\bar X_{\mathrm{trt},1}+\bar X_{\mathrm{trt},2}),
\]
and similarly for $\bar X_{\mathrm{ctrl}}$, giving 
$T=\frac12(T_1+T_2)$ with $T_s:=\bar X_{\mathrm{trt},s}-\bar X_{\mathrm{ctrl},s}$.
Since Strategy~2 permutes within each stratum independently, 
$T_1$ and $T_2$ are conditionally independent given $X$.

\medskip
\noindent\emph{Step 2: within-stratum variance.}
Applying Proposition~\ref{prop:var-full} to Stratum~$s$
(with $n\to N_s=n/2$ and $n_1\to n/4$) gives
\[
\Var(T_s\mid X)=\frac{4Q_s}{(n/2)(n/2-1)},
\qquad 
Q_s:=\sum_{i\in\mathcal S_s}(X_i-\bar X_s)^2.
\]
Combining via conditional independence,
\[
\Var_{\mathrm{strat}}(T\mid X)
=\frac14\bigl[\Var(T_1\mid X)+\Var(T_2\mid X)\bigr]
=\frac{Q_1+Q_2}{(n/2)(n/2-1)}.
\]

\medskip
\noindent\emph{Step 3: limit of $Q_s/(n/2-1)$ under $H_1$.}
Within Stratum~$s$, $X_i=\mu_s+(\delta/\sqrt n)w_{0,i}+\varepsilon_i$ and
$\bar X_s=\mu_s+\delta/(2\sqrt n)+\bar\varepsilon_s$, 
so the stratum effect cancels:
\[
X_i-\bar X_s=\frac{\delta}{\sqrt n}\Big(w_{0,i}-\tfrac12\Big)+(\varepsilon_i-\bar\varepsilon_s).
\]
Expanding $Q_s=\sum_{i\in\mathcal S_s}(X_i-\bar X_s)^2$ produces three terms.
The noise term $\sum_{i\in\mathcal S_s}(\varepsilon_i-\bar\varepsilon_s)^2/(N-1)\toP 1$ 
by the law of large numbers (with $N=n/2$).
The signal-squared term equals 
$\frac{\delta^2}{n}\sum_{i\in\mathcal S_s}(w_{0,i}-1/2)^2=\frac{\delta^2}{n}\cdot\frac{N}{4}=\frac{\delta^2}{8}$, 
whose contribution to $Q_s/(N-1)$ is $\delta^2/[8(N-1)]\to 0$.
For the cross term, $\sum_{i\in\mathcal S_s}(w_{0,i}-1/2)=0$
(exactly $N/2$ treated in each stratum), so
$\sum_{i\in\mathcal S_s}(w_{0,i}-1/2)(\varepsilon_i-\bar\varepsilon_s)
=\sum_{i\in\mathcal S_s}(w_{0,i}-1/2)\varepsilon_i$, 
a mean-zero sum of independent $\varepsilon_i$ with variance $N/4=O(n)$, 
hence $O_P(\sqrt n)$ by Chebyshev.
Dividing by $\sqrt n(N-1)$ gives $O_P(1/n)\to 0$.
Thus $Q_s/(n/2-1)\toP 1$ and 
$n\cdot\Var_{\mathrm{strat}}(T\mid X)=\frac{n(Q_1+Q_2)}{(n/2)(n/2-1)}\toP 4$.
\end{proof}

\subsection{Proof of Proposition~\ref{prop:power-compare}}\label{app:proof-power-compare}

The proof requires the following intermediate result.

\begin{lemma}[Permutation CLT for each strategy under $H_1$]\label{lem:perm-CLT-strat}
Under $H_1$, write $\hat\sigma_{\mathrm{full}}^2:=\Var_{\mathrm{full}}(T\mid X)$ 
and $\hat\sigma_{\mathrm{strat}}^2:=\Var_{\mathrm{strat}}(T\mid X)$.
Then both strategies produce asymptotically Gaussian permutation distributions:
\begin{enumerate}[label=(\roman*)]
\item $\sup_t|\mP(T/\hat\sigma_{\mathrm{full}}\le t\mid X)-\Phi(t)|\toP 0$.
\item $\sup_t|\mP(T/\hat\sigma_{\mathrm{strat}}\le t\mid X)-\Phi(t)|\toP 0$.
\end{enumerate}
Consequently,
$c_{\mathrm{full}}(\alpha\mid X)=z_{1-\alpha}\,\hat\sigma_{\mathrm{full}}(1+o_P(1))$
and $c_{\mathrm{strat}}(\alpha\mid X)=z_{1-\alpha}\,\hat\sigma_{\mathrm{strat}}(1+o_P(1))$.
\end{lemma}

\begin{proof}
\emph{Part (i).}
The permutation CLT for simple random sampling 
(\citet{hajek1961}; \citet{hajek1999theory}) requires
\[
\frac{\max_{1\le i\le n}(X_i-\bar X)^2}{\sum_{i=1}^n(X_i-\bar X)^2}\toP 0.
\]
Under $H_1$ with Gaussian noise, the denominator satisfies 
$\frac1n\sum_i(X_i-\bar X)^2\toP 1+\mu^2>0$ (Proposition~\ref{prop:var-full}),
while $\max_i(X_i-\bar X)^2=O_P(\log n)$ by standard Gaussian maximum bounds.
The ratio is therefore $O_P(\log n/n)\to 0$, so the permutation CLT applies, giving~(i).

\medskip
\noindent\emph{Part (ii).}
By Proposition~\ref{prop:var-strat} and Step~1 of its proof, 
$T=\frac12(T_1+T_2)$ with $T_1,T_2$ conditionally independent given $X$ 
and within-stratum variances 
$\hat\sigma_{\mathrm{strat},s}^2:=\Var(T_s\mid X)=4Q_s/[(n/2)(n/2-1)]$.
Within Stratum~$s$, $Q_s/(n/2-1)\toP 1$ and 
$\max_{i\in\mathcal S_s}(X_i-\bar X_s)^2=O_P(\log n)$, 
so the permutation CLT applies within each stratum:
$T_s/\hat\sigma_{\mathrm{strat},s}\mid X\dto\N(0,1)$.
By Step~3 of the previous proof, $n\hat\sigma_{\mathrm{strat},s}^2\toP 8$ for $s=1,2$, 
while $n\hat\sigma_{\mathrm{strat}}^2\toP 4$. Hence
\[
\frac{\hat\sigma_{\mathrm{strat},s}^2}{\hat\sigma_{\mathrm{strat}}^2}\toP 2,
\qquad
\left(\frac{\hat\sigma_{\mathrm{strat},s}}{2\hat\sigma_{\mathrm{strat}}}\right)^2\toP \frac12,
\quad s=1,2.
\]
Writing
\[
\frac{T}{\hat\sigma_{\mathrm{strat}}}
=\frac{T_1}{\hat\sigma_{\mathrm{strat},1}}\cdot\frac{\hat\sigma_{\mathrm{strat},1}}{2\hat\sigma_{\mathrm{strat}}}
+\frac{T_2}{\hat\sigma_{\mathrm{strat},2}}\cdot\frac{\hat\sigma_{\mathrm{strat},2}}{2\hat\sigma_{\mathrm{strat}}},
\]
Slutsky's theorem yields $T/\hat\sigma_{\mathrm{strat}}\dto\N(0,1)$.

\medskip
\noindent\emph{Critical-value statements.}
If the standardized permutation cdf converges uniformly to $\Phi$, 
the quantile function converges pointwise, giving 
$c(\alpha\mid X)/\hat\sigma\toP z_{1-\alpha}$
and hence $c(\alpha\mid X)=z_{1-\alpha}\,\hat\sigma\,(1+o_P(1))$.
\end{proof}

\begin{proof}[Proof of Proposition~\ref{prop:power-compare}]
\emph{Step 1: distribution of $T$ under $H_1$.}
With balanced design (each arm has $n/4$ units from each stratum), 
the stratum effects cancel exactly in $\bar X_{\mathrm{trt}}-\bar X_{\mathrm{ctrl}}$:
\[
\bar X_{\mathrm{trt}}=\underbrace{\frac{(n/4)\mu+(n/4)(-\mu)}{n/2}}_{=0}+\frac{\delta}{\sqrt n}+\bar\varepsilon_{\mathrm{trt}},
\qquad
\bar X_{\mathrm{ctrl}}=0+0+\bar\varepsilon_{\mathrm{ctrl}},
\]
so $T=\delta/\sqrt n+(\bar\varepsilon_{\mathrm{trt}}-\bar\varepsilon_{\mathrm{ctrl}})$.
Since $\bar\varepsilon_{\mathrm{trt}}$ and $\bar\varepsilon_{\mathrm{ctrl}}$ are independent 
(disjoint units) with variance $1/(n/2)=2/n$ each, their difference has variance $4/n$.
Under Gaussian noise, $\sqrt n T/2=\delta/2+(\sqrt n/2)(\bar\varepsilon_{\mathrm{trt}}-\bar\varepsilon_{\mathrm{ctrl}})$
is exactly $\N(\delta/2,1)$ in finite samples, hence $Z_n:=\sqrt n T/2\dto\N(\delta/2,1)$.

\medskip
\noindent\emph{Step 2: asymptotic critical values.}
By Lemma~\ref{lem:perm-CLT-strat}, 
$c_{\mathrm{full}}(\alpha\mid X)=z_{1-\alpha}\,\hat\sigma_{\mathrm{full}}(1+o_P(1))$
and $c_{\mathrm{strat}}(\alpha\mid X)=z_{1-\alpha}\,\hat\sigma_{\mathrm{strat}}(1+o_P(1))$.
Propositions~\ref{prop:var-full}--\ref{prop:var-strat} give 
$n\hat\sigma_{\mathrm{full}}^2\toP 4(1+\mu^2)$ and $n\hat\sigma_{\mathrm{strat}}^2\toP 4$, so
\[
\frac{\sqrt n\cdot c_{\mathrm{full}}(\alpha\mid X)}{2}\toP z_{1-\alpha}\sqrt{1+\mu^2},
\qquad
\frac{\sqrt n\cdot c_{\mathrm{strat}}(\alpha\mid X)}{2}\toP z_{1-\alpha}.
\]

\medskip
\noindent\emph{Step 3: limiting power.}
The rejection events $\{T>c(\alpha\mid X)\}$ equal $\{Z_n>\sqrt n\,c(\alpha\mid X)/2\}$.
Combining Steps~1--2 with Slutsky's theorem yields
\begin{align*}
\mP(T>c_{\mathrm{full}}(\alpha\mid X))
&\to 1-\Phi\Big(z_{1-\alpha}\sqrt{1+\mu^2}-\tfrac{\delta}{2}\Big),\\
\mP(T>c_{\mathrm{strat}}(\alpha\mid X))
&\to 1-\Phi\Big(z_{1-\alpha}-\tfrac{\delta}{2}\Big).
\end{align*}

\medskip
\noindent\emph{Step 4: power ordering.}
For $\alpha<1/2$, $z_{1-\alpha}>0$.
When $\mu>0$, $\sqrt{1+\mu^2}>1$, so 
$z_{1-\alpha}\sqrt{1+\mu^2}>z_{1-\alpha}$ and thus 
$\mathrm{Power}_{\mathrm{strat}}>\mathrm{Power}_{\mathrm{full}}$.
As $\mu\to\infty$, $z_{1-\alpha}\sqrt{1+\mu^2}\to\infty$ and 
$\mathrm{Power}_{\mathrm{full}}\to 1-\Phi(\infty)=0$,
while $\mathrm{Power}_{\mathrm{strat}}$ is independent of $\mu$.
\end{proof}

\end{document}